\documentclass[journal]{IEEEtran}

\usepackage[utf8]{inputenc}
\usepackage{color}
\usepackage{xcolor}
\usepackage{array}
\usepackage{verbatim}
\usepackage{float}
\usepackage{amsmath}
\usepackage{amsthm}
\usepackage{amssymb}
\usepackage{graphicx}
\usepackage{longtable}
\usepackage{multirow}
\usepackage{booktabs}
\usepackage[unicode=true,
bookmarks=false,
breaklinks=false,pdfborder={0 0 1},colorlinks=false]
{hyperref}
\hypersetup{
	colorlinks,bookmarksopen,bookmarksnumbered,citecolor=blue,urlcolor=blue}
\usepackage{cite}

\usepackage{lipsum}
\usepackage{mathtools}
\usepackage{cuted}
\providecommand{\tabularnewline}{\\}
\usepackage{algorithmic}
\usepackage{longtable}

\floatstyle{ruled}
\newfloat{algorithm}{tbp}{loa}
\providecommand{\algorithmname}{Algorithm}
\floatname{algorithm}{\protect\algorithmname}

\makeatletter
\let\oldforeign@language\foreign@language
\DeclareRobustCommand{\foreign@language}[1]{%
	\lowercase{\oldforeign@language{#1}}}

\let\oldforeign@language\foreign@language
\DeclareRobustCommand{\foreign@language}[1]{%
	\lowercase{\oldforeign@language{#1}}}

\ifCLASSINFOpdf
\else
\fi

\@ifundefined{showcaptionsetup}{}{%
	\PassOptionsToPackage{caption=false}{subfig}}
\usepackage{subfig}

\newcommand{\MYfooter}{\smash{
		\hfil\parbox[t][\height][t]{\textwidth}{\centering
			\thepage}\hfil\hbox{}}}

\def\ps@IEEEtitlepagestyle{%
	\def\@oddhead{\parbox[t][\height][t]{\textwidth}{\centering \scriptsize
			Personal use of this material is permitted. Permission from the author(s) and/or copyright holder(s), must be obtained for all other uses. Please contact us and provide details if you believe this document breaches copyrights.\\
			\noindent\makebox[\linewidth]{}
		}\hfil\hbox{}}%
	\def\@evenhead{\scriptsize\thepage \hfil \leftmark\mbox{}}%
	\def\@oddfoot{\parbox[t][\height][l]{\textwidth}{
			\vspace{-20pt}{\rule{\textwidth}{0.4pt}}\\ \footnotesize\underline{To cite this article:}
			{\bf{\footnotesize\textcolor{red}{H. A. Hashim, A. E. E. Eltoukhy, and K. G. Vamvoudakis, "UWB Ranging and IMU Data Fusion: Overview and Nonlinear Stochastic Filter for Inertial Navigation," IEEE Transactions on Intelligent Transportation Systems, pp. 1-11, 2023.}}} doi: \href{https://doi.org/10.1109/TITS.2023.3309288}{10.1109/TITS.2023.3309288}\\
			\noindent\makebox[\linewidth]
		}\hfil\hbox{}}%
	\def\@evenfoot{\MYfooter}}

\makeatother
\newtheorem{defn}{Definition}

\newtheorem{lem}{Lemma}

\newtheorem{thm}{Theorem}

\newtheorem{assum}{Assumption}

\begin{document}
	\bstctlcite{IEEEexample:BSTcontrol}

	\title{UWB Ranging and IMU Data Fusion: Overview and Nonlinear Stochastic Filter for Inertial Navigation}

\author{Hashim A. Hashim, Abdelrahman E. E. Eltoukhy, and Kyriakos G. Vamvoudakis
	\thanks{This work was supported in part by the National Sciences and Engineering Research Council of Canada (NSERC) under the grants RGPIN-2022-04937 and by the National Science Foundation under grant Nos. S\&AS-1849264, CPS-1851588, and CPS-2038589.}
	\thanks{H. A. Hashim is with the Department of Mechanical and Aerospace Engineering, Carleton University, Ottawa, ON, K1S 5B6, Canada (e-mail: hhashim@carleton.ca).
		A. E.E. Eltoukhy is with the Department of Industrial and Systems Engineering, The Hong Kong Polytechnic University, Hung Hum, Hong Kong (e-mail: abdelrahman.eltoukhy@polyu.edu.hk).	
		K. G. Vamvoudakis is with the Daniel Guggenheim School of Aerospace Engineering, Georgia Institute of Technology, Atlanta, GA, 30332, USA (e-mail: kyriakos@gatech.edu).}
}



\maketitle

\begin{abstract}
This paper proposes a nonlinear stochastic complementary filter design
for inertial navigation that takes advantage of a fusion of Ultra-wideband
(UWB) and Inertial Measurement Unit (IMU) technology ensuring semi-global
uniform ultimate boundedness (SGUUB) of the closed loop error signals
in mean square. The proposed filter estimates the vehicle's orientation,
position, linear velocity, and noise covariance. The filter is designed
to mimic the nonlinear navigation motion kinematics and is posed on
a matrix Lie Group, the extended form of the Special Euclidean Group
$\mathbb{SE}_{2}\left(3\right)$. The Lie Group based structure of
the proposed filter provides unique and global representation avoiding
singularity (a common shortcoming of Euler angles) as well as non-uniqueness
(a common limitation of unit-quaternion). Unlike Kalman-type filters,
the proposed filter successfully addresses IMU measurement noise considering
unknown upper-bounded covariance. Although the navigation estimator
is proposed in a continuous form, the discrete version is also presented.
Moreover, the unit-quaternion implementation has been provided in
the Appendix. Experimental validation performed using a publicly available
real-world six-degrees-of-freedom (6 DoF) flight dataset obtained
from an unmanned Micro Aerial Vehicle (MAV) illustrating the robustness
of the proposed navigation technique.
\end{abstract}


\IEEEpeerreviewmaketitle{}

\section{Introduction}

\IEEEPARstart{I}{nertial} navigation is a fundamental robotic task commonly accomplished
by fusing information from multiple sensors \cite{beauvisage2021robust,hashim2021_COMP_ENG_PRAC,zou2021comparative,hashim2022ExpVTOL,liu2022formation}.
Traditionally, outdoor robotic missions carried out by Unmanned Aerial
Vehicles (UAVs), mobile robots, and ground vehicles utilize a combination
of Global Positioning Systems (GPS) and an Inertial Measurement Unit
(IMU) to extract the navigation components, namely vehicle's orientation
(commonly known as attitude), position, and linear velocity, essential
for the success of any control missions \cite{d2018uav}. However,
GPS signal is susceptible to obstructions, multipath, fading, and/or
denial (indoor mission and in harsh weather). Given the possibility
of GPS signal loss, availability of a back-up technology for accurate
estimation of the navigation components is crucial to prevent mission
failure until GPS signal is restored. As such, the research community
has been actively seeking to address the above challenge using different
Inertial Navigation Systems (INSs). For instance, vision-based aided
navigation techniques that rely on a vision unit (monocular or stereo
camera) and an IMU have been employed \cite{hashim2021_COMP_ENG_PRAC,beauvisage2021robust,cheng2014seamless,hashim2023_ISA,hashim2021ACC}.
Other researchers have integrated a Light Detection and Ranging (LiDAR)
sensor and an IMU \cite{zou2021comparative}. Navigation components
estimation can also be enabled by fusing an Ultra-wideband (UWB) and
an IMU \cite{liu2022formation,tian2019resetting,xiong2021adaptive}.
In comparison with LiDAR and vision units, UWB and IMU fusion reduces
the cost, size, and weight of the sensing unit and the power requirements.
Moreover, UWB enables positioning whether the communication between
the tag (attached to the robot) and the fixed anchors (source) is
within the Line-of-sight (LOS) or Non-line-of-sight (NLOS). However,
the main challenge of using a combination of UWB and IMU is their
proneness to measurement uncertainties \cite{dwek2019improving,guo2019ultra,shen2010fundamental,navratil2022concurrent}.

Localization of a rigid-body (e.g., UAVs and ground vehicles) aims
to fully define attitude and position \cite{hashim2021_COMP_ENG_PRAC}.
A vehicle equipped with a 9-axis IMU (composed of a gyroscope, an
accelerometer, and a magnetometer) allows for attitude determination.
A vehicle equipped with a tag accessed by a group of fixed anchors
allows for position determination. Attitude and position determination
are normally impaired by noise, and therefore, require a robust filter
to (1) attenuate the noise effect, (2) estimate the hidden states
(the linear velocity), and (3) produce reasonable navigation estimates.
Over the last few years, several Gaussian navigation filters have
been proposed based on the fusion of UWB and IMU to achieve higher
estimation accuracy and reduce measurement noise. Examples of IMU-UWB-based
Gaussian navigation filters include a Kalman Filter (KF) that has
been developed for indoor localization systems \cite{feng2020kalman},
a tightly coupled Extended Kalman Filter (EKF) that addresses the
divergence of KFs \cite{wen2020new,strohmeier2018ultra}, and an Unscented
Kalman Filter (UKF) suited for a group of unmanned ground vehicles
that uses a set of sigma points to improve the probability distribution
\cite{liu2022formation}. In addition, Particle filters (PFs) reliant
on IMU-UWB-fusion have been introduced \cite{tian2019resetting}.
PFs do not follow the Gaussian assumption and they are typically classified
as stochastic filters \cite{hashim2018SO3Stochastic}. The main shortcoming
of KF, EKF, and UKF, is the fact that they are based on optimal minimum-energy
which is first order adopting linearization around a nominal point
\cite{hashim2018SO3Stochastic,chui2017kalman}. As such, higher order
terms are disregarded resulting in degradation of the estimation accuracy.
Furthermore, UKF sigma points add complexity to the implementation
process \cite{hashim2018SO3Stochastic}. The limitations of PFs are
two-fold: 1) they have higher computational cost (unfit for small
scale systems) \cite{hashim2018SO3Stochastic}, and 2) the stability
results are not indicative of the proximity of the solution to the
optimal one \cite{zamani2013minimum,chui2017kalman}. In addition,
the KF, EKF, UKF and PF techniques utilize Euler angles which are
subject to singularities. It is critical to consider that the navigation
kinematics of a vehicle moving in three-dimensional (3D) space are
highly nonlinear. As such, to capture the nonlinearity and address
the singularity issue of Euler angles, the navigation problem is best
modeled on the Lie Group. Moreover, the impracticality of the known
noise covariance assumption adopted by Kalman-type filters and PFs
has to be addressed \cite{hashim2018SO3Stochastic,crassidis2007survey}.

Motivated by the advantages and limitations of the above literature
discussion, this work aims to capture the complex nonlinear nature
of the navigation kinematics. Hence, the problem is modeled on the
Lie Group of the extended form of the Special Euclidean Group $\mathbb{SE}_{2}\left(3\right)$.
The navigation approach proposed in this work makes provision for
the uncertainty in IMU measurements and considers a vehicle equipped
with a UWB tag and availability of $n$ UWB anchors. Consequently,
the contributions of this work are as follows:
\begin{enumerate}
	\item[1)] A nonlinear stochastic complementary filter for inertial navigation
	developed on the Lie Group of $\mathbb{SE}_{2}\left(3\right)$ reliant
	on the direct UWB and IMU measurements is proposed. 
	\item[2)] The filter is characterized with a geometric framework, able to preserve
	the Lie Group and avoid singularity unlike Gaussian navigation filters
	in \cite{feng2020kalman,wen2020new,strohmeier2018ultra,liu2022formation}.
	\item[3)] The stochastic filter effectively addresses unknown measurement noise
	introduced by an IMU.
	\item[4)] The proposed filter guarantees semi-global uniform ultimate boundedness
	(SGUUB) of the closed loop error signals in mean square using Lyapunov
	stability.
\end{enumerate}
The paper is composed of seven Sections and an Appendix. Section \ref{sec:Preliminaries-and-Math}
presents preliminaries and the related math notation. Section \ref{sec:SE3_Problem-Formulation}
details the UWB positioning problem, orientation determination using
IMU, and the navigation estimation problem. Section \ref{sec:UWB_Filter}
presents the proposed nonlinear stochastic navigation filter design
on $\mathbb{SE}_{2}\left(3\right)$. Section \ref{sec:UWB_Pose} illustrates
the possibility of position and orientation determination using solely
UWB technology conditioned on potential advancement of UWB ranging
accuracy. Section \ref{sec:UWB_Simulations} demonstrates the robustness
of the proposed approach by means of testing it on a real-world dataset.
Finally, Section \ref{sec:SE3_Conclusion} concludes the work.

\section{Preliminaries\label{sec:Preliminaries-and-Math}}

\begin{table}[t]
	\centering{}\caption{\label{tab:Table-of-Notations2}Nomenclature}
	\begin{tabular}{ll>{\raggedright}p{6.3cm}}
		\toprule 
		\addlinespace[0.1cm]
		$\mathbb{R}^{n\times m}$ & : & $n$-by-$m$ real dimensional space\tabularnewline
		\addlinespace[0.1cm]
		$\mathbb{SO}\left(3\right)$ & : & Special Orthogonal Group\tabularnewline
		\addlinespace[0.1cm]
		$\mathfrak{so}\left(3\right)$ & : & Lie-algebra of $\mathbb{SO}\left(3\right)$\tabularnewline
		\addlinespace[0.1cm]
		$\mathbb{SE}_{2}\left(3\right)$ & : & Extended Special Euclidean Group, $\mathbb{SE}_{2}\left(3\right)=\mathbb{SO}\left(3\right)\times\mathbb{R}^{3}\times\mathbb{R}^{3}$\tabularnewline
		\addlinespace[0.1cm]
		$\mathbb{S}^{3}$ & : & \noindent Three-unit-sphere\tabularnewline
		\addlinespace[0.1cm]
		$h_{i}$ & : & $i$th UWB Anchor (source) position, $h_{i}\in\mathbb{R}^{3}$\tabularnewline
		\addlinespace[0.1cm]
		$P$ & : & Unknown vehicle (UWB Tag) position, $P\in\mathbb{R}^{3}$\tabularnewline
		\addlinespace[0.1cm]
		$\hat{P}$ & : & Estimated vehicle position, $\hat{P}\in\mathbb{R}^{3}$\tabularnewline
		\addlinespace[0.1cm]
		$R$ and $\hat{R}$ & : & True (unknown) and estimated attitude, $R,\hat{R}\in\mathbb{SO}\left(3\right)$\tabularnewline
		\addlinespace[0.1cm]
		$V$ and $\hat{V}$ & : & True (unknown) and estimated linear velocity, $V,\hat{V}\in\mathbb{R}^{3}$\tabularnewline
		\addlinespace[0.1cm]
		$\Omega$ and $\Omega_{m}$ & : & True and measured angular velocity, $\Omega,\Omega_{m}\in\mathbb{R}^{3}$\tabularnewline
		\addlinespace[0.1cm]
		$a$ and $a_{m}$ & : & True and measured acceleration, $a,a_{m}\in\mathbb{R}^{3}$\tabularnewline
		\addlinespace[0.1cm]
		$X$ and $\hat{X}$ & : & True (unknown) and estimated navigation, $X,\hat{X}\in\mathbb{SE}_{2}\left(3\right)$\tabularnewline
		\addlinespace[0.1cm]
		$P_{y}$ & : & Reconstructed position, $R_{y}\in\mathbb{R}^{3}$\tabularnewline
		\addlinespace[0.1cm]
		$\tilde{R}$ & : & Attitude estimation error, $\tilde{R}\in\mathbb{SO}\left(3\right)$\tabularnewline
		\addlinespace[0.1cm]
		$\tilde{P}$ and $\tilde{V}$ & : & Position and linear velocity estimation error, $\tilde{P},\tilde{V}\in\mathbb{R}^{3}$\tabularnewline
		\bottomrule
	\end{tabular}
\end{table}

Throughout this paper $\mathbb{R}_{+}$ refers to the set of nonnegative
real numbers. For $v\in\mathbb{R}^{n}$ and $Q\in\mathbb{R}^{n\times m}$,
the Euclidean norm of $v$ is described by $||v||=\sqrt{v^{\top}v}$
and the Frobenius norm of $Q$ is defined by $||Q||_{F}=\sqrt{{\rm Tr}\{QQ^{*}\}}$
where $*$ is a conjugate transpose. $\mathbf{I}_{n}$ represents
an $n$-by-$n$ identity matrix and $0_{n\times m}$ describes an
$n$-by-$m$ zero matrix. For $M_{r}\in\mathbb{R}^{n\times n}$, $\lambda(M_{r})=\{\lambda_{1},\lambda_{2},\ldots,\lambda_{n}\}$
describes the set of eigenvalues of $M_{r}$ with $\overline{\lambda}_{M_{r}}=\overline{\lambda}(M_{r})$
referring to the maximum eigenvalue and $\underline{\lambda}_{M_{r}}=\underline{\lambda}(M_{r})$
being the minimum eigenvalue of $M_{r}$. ${\rm det}(\cdot)$ denotes
a determinant, $\exp(\cdot)$ represents exponential, and $\mathbb{E}[\cdot]$
refers to an expected value. $\left\{ \mathcal{I}\right\} $ describes
fixed inertial-frame and $\left\{ \mathcal{B}\right\} $ represents
body-frame fixed to the navigating vehicle. The Special Orthogonal
Group $\mathbb{SO}\left(3\right)$ is described by \cite{hashim2018SO3Stochastic,hashim2019SO3Wiley}
\[
\mathbb{SO}\left(3\right)=\left\{ \left.R\in\mathbb{R}^{3\times3}\right|R^{\top}R=RR^{\top}=\mathbf{I}_{3}\text{, }{\rm det}\left(R\right)=+1\right\} 
\]
where $R\in\mathbb{SO}\left(3\right)$ refers to vehicle's orientation
(attitude) in the $\left\{ \mathcal{B}\right\} $-frame. The Lie algebra
of $\mathbb{SO}(3)$ is denoted as $\mathfrak{so}(3)$ such that
\begin{align*}
	\mathfrak{so}\left(3\right) & =\{[v]_{\times}\in\mathbb{R}^{3\times3}|[v]_{\times}^{\top}=-[v]_{\times},v\in\mathbb{R}^{3}\}\\
	\left[v\right]_{\times} & =\left[\begin{array}{ccc}
		0 & -v_{3} & v_{2}\\
		v_{3} & 0 & -v_{1}\\
		-v_{2} & v_{1} & 0
	\end{array}\right]\in\mathfrak{so}\left(3\right),\hspace{1em}v=\left[\begin{array}{c}
		v_{1}\\
		v_{2}\\
		v_{3}
	\end{array}\right]
\end{align*}
where $[v]_{\times}$ describes a skew symmetric matrix{\small{}.}
$\mathbf{vex}:\mathfrak{so}\left(3\right)\rightarrow\mathbb{R}^{3}$
represents the inverse mapping of $[\cdot]_{\times}$ such that $\mathbf{vex}([v]_{\times})=v,\forall v\in\mathbb{R}^{3}$.
The anti-symmetric projection on $\mathfrak{so}\left(3\right)$ is
described by $\boldsymbol{\mathcal{P}}_{a}(M_{r})=\frac{1}{2}(M_{r}-M_{r}^{\top})\in\mathfrak{so}\left(3\right),\forall M_{r}\in\mathbb{R}^{3\times3}$.
For $M_{r}\in\mathbb{R}^{3\times3}$, $\boldsymbol{\Upsilon}=\mathbf{vex}\circ\boldsymbol{\mathcal{P}}_{a}$
represents a composition mapping with $\boldsymbol{\Upsilon}(M_{r})=\mathbf{vex}(\boldsymbol{\mathcal{P}}_{a}(M_{r}))\in\mathbb{R}^{3}$.
The Euclidean distance of thr vehicle's orientation $R\in\mathbb{SO}\left(3\right)$
is defined by
\begin{equation}
	||R||_{{\rm I}}={\rm Tr}\{\mathbf{I}_{3}-R\}/4\in\left[0,1\right]\label{eq:NAV_Ecul_Dist}
\end{equation}
with $-1\leq{\rm Tr}\{R\}\leq3$ and $||R||_{{\rm I}}=\frac{1}{8}||\mathbf{I}_{3}-R||_{F}^{2}$,
refer to \cite{hashim2018SO3Stochastic}. Likewise, we define $||MR||_{{\rm I}}={\rm Tr}\{M-MR\}/4$
for all $M\in\mathbb{R}^{3\times3}$. Consider a vehicle navigating
in 3D space where $R\in\mathbb{SO}\left(3\right)$, $P\in\mathbb{R}^{3}$,
and $V\in\mathbb{R}^{3}$ denote its attitude, position, and velocity,
respectively, with $R\in\{\mathcal{B}\}$ and $P,V\in\{\mathcal{I}\}$.
Define $\mathbb{SE}_{2}\left(3\right)=\mathbb{SO}\left(3\right)\times\mathbb{R}^{3}\times\mathbb{R}^{3}\subset\mathbb{R}^{5\times5}$
\cite{barrau2016invariant} add later as the extended form of the
Special Euclidean Group $\mathbb{SE}\left(3\right)=\mathbb{SO}\left(3\right)\times\mathbb{R}^{3}\subset\mathbb{R}^{4\times4}$
with
\begin{align}
	\mathbb{SE}_{2}\left(3\right) & =\{\left.X\in\mathbb{R}^{5\times5}\right|R\in\mathbb{SO}\left(3\right),P,V\in\mathbb{R}^{3}\}\label{eq:NAV_SE2_3}\\
	X=\Psi( & R,P,V)=\left[\begin{array}{ccc}
		R & P & V\\
		0_{1\times3} & 1 & 0\\
		0_{1\times3} & 0 & 1
	\end{array}\right]\in\mathbb{SE}_{2}\left(3\right)\label{eq:NAV_X}
\end{align}
$X\in\mathbb{SE}_{2}\left(3\right)$ is known as a homogeneous navigation
matrix. $T_{X}\mathbb{SE}_{2}\left(3\right)\in\mathbb{R}^{5\times5}$
is the tangent space of $\mathbb{SE}_{2}\left(3\right)$ at point
$X$. Let us introduce a submanifold $\mathcal{U}_{\mathcal{M}}=\mathfrak{so}\left(3\right)\times\mathbb{R}^{3}\times\mathbb{R}^{3}\times\mathbb{R}\subset\mathbb{R}^{5\times5}$
where
\begin{align}
	\mathcal{U}_{\mathcal{M}} & =\left\{ \left.u([\Omega\text{\ensuremath{]_{\times}}},V,a,\varepsilon)\right|[\Omega\text{\ensuremath{]_{\times}}}\in\mathfrak{so}(3),V,a\in\mathbb{R}^{3},\varepsilon\in\mathbb{R}\right\} \nonumber \\
	u( & [\Omega\text{\ensuremath{]_{\times}}},V,a,\varepsilon)=\left[\begin{array}{ccc}
		[\Omega\text{\ensuremath{]_{\times}}} & V & a\\
		0_{1\times3} & 0 & 0\\
		0_{1\times3} & \varepsilon & 0
	\end{array}\right]\in\mathcal{U}_{\mathcal{M}}\subset\mathbb{R}^{5\times5}\label{eq:NAV_u}
\end{align}
such that $\Omega\in\mathbb{R}^{3}$, $V\in\mathbb{R}^{3}$, and $a\in\mathbb{R}^{3}$
refer to vehicle's true angular velocity, linear velocity, and apparent
acceleration respectively. Note that $\Omega,a\in\{\mathcal{B}\}$.
For more details on $\mathbb{SE}_{2}\left(3\right)$ and $\mathcal{U}_{\mathcal{M}}$
visit \cite{hashim2021_COMP_ENG_PRAC,hashim2021ACC}. Let $y\in{\rm \mathbb{R}}^{3}$,
$M\in\mathbb{R}^{3\times3}$, and $R\in\mathbb{SO}\left(3\right)$.
The following identities hold:
\begin{align}
	\left[Ry\right]_{\times}= & R\left[y\right]_{\times}R^{\top}\label{eq:UWB_Identity1}\\
	{\rm Tr}\{M\left[y\right]_{\times}\}= & {\rm Tr}\{\boldsymbol{\mathcal{P}}_{a}(M)\left[y\right]_{\times}\}=-2\mathbf{vex}(\boldsymbol{\mathcal{P}}_{a}(M))^{\top}y\label{eq:UWB_Identity2}
\end{align}

\section{Problem Formulation\label{sec:SE3_Problem-Formulation}}

UWB technology has several notable advantages making it an excellent
candidate for a variety of applications. The UWB signal has large
bandwidth with short life-time (frequency is inversely proportional
with time) which results in reasonable positioning accuracy \cite{hashim2023nonlinear,dwek2019improving}.
The distinguishing feature of the UWB signal is its short wavelength
allowing it to be robust against multipath interference and fading
unlike GPS signal. has UWB low power consumption and fast communication
speed. Additionally, the UWB signals can penetrate obstacles providing
localization in LOS and NLOS. New technologies shows UWB precision
of approximately 10 centimeters within ranging distance of 100 meters.
Finally, UWB technology hardware is compact and allows for low-cost
implementation \cite{hashim2023nonlinear,dwek2019improving,guo2019ultra,shen2010fundamental}.
However, similarly to IMU, a significant limitation of UWB is a high
level of measurement noise. To gain understanding of UWB technology,
let us define terms ``anchor'' and ``tag''. An anchor refers to
a fixed UWB sensor with a known location, while a tag stands for a
UWB attached to a navigating vehicle. A tag generally has to exchange
signals with multiple anchors to determine its position in 3D space.

\subsection{UWB and Time of Arrival\label{subsec:TOA}}

Time of Arrival (TOA) is a well known approach that uses instant time
between transmitter and receiver (travel time) to provide range or
distance between anchor and tag, as long as the tag is in the range
of the transmitted signal \cite{hua2016joint}. Let $d_{i}$ denote
the $i$th distance between the tag and the $i$th anchor. The left
portion of Fig. \ref{fig:AnchRange}.(a) illustrates maximum range
of an anchor and the distance between an anchor and the tag. 

Unique position determination of the tag in 2D space requires a minimum
of 3 anchors. As illustrated by the left portion of Fig. \ref{fig:AnchRange}.(a),
knowledge of the range of one anchor allows to identify a circle where
the tag could be potentially located. Two anchors narrow the potential
location of the tag down to two options (left portion of Fig. \ref{fig:AnchRange}.(a)).
Finally, introducing the third anchor allows to precisely pinpoint
the position of the tag in 2D space as illustrated in Fig. \ref{fig:AnchRange}.(b).
By extension, a minimum of 4 anchors is required to uniquely position
the tag in 3D space.

\begin{figure}[h]
	\begin{centering}
		\subfloat[Anchor maximum range and distance to the tag.]{\begin{centering}
				\centering\includegraphics[scale=0.33]{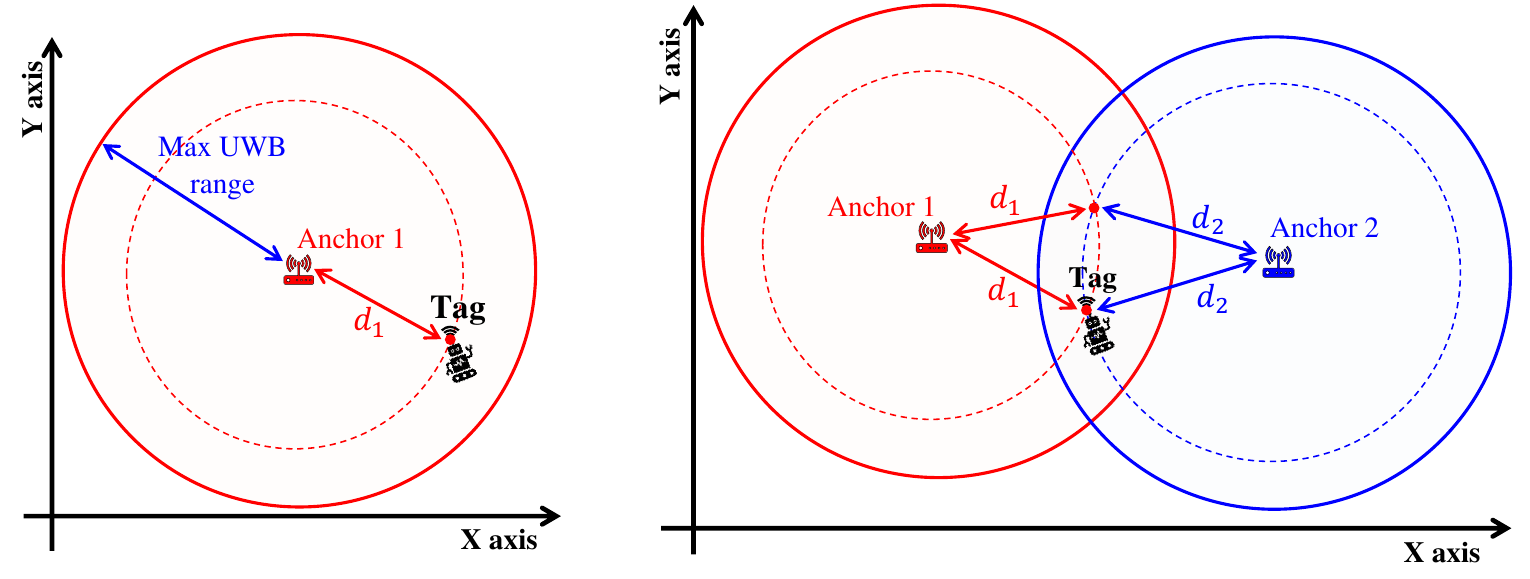}
				\par\end{centering}
		}
		\par\end{centering}
	\centering{}\subfloat[Positioning in 2D space.]{\begin{centering}
			\centering\includegraphics[scale=0.38]{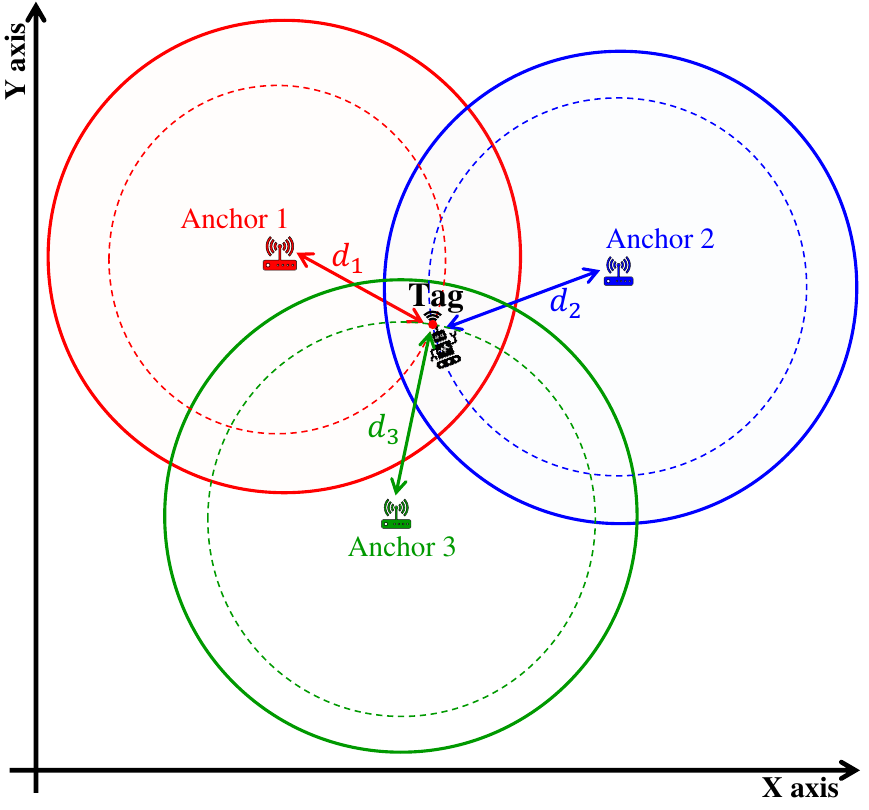}
			\par\end{centering}
	}\caption{\label{fig:AnchRange} Positioning with UWB.}
\end{figure}

Let us define the range $d_{i}$ in 3D space between the fixed $i$th
anchor positioned at $h_{i}=[x_{i},y_{i},z_{i}]^{\top}\in\mathbb{R}^{3}$
and the moving vehicle (UWB tag position) located at $P=[x,y,z]^{\top}\in\mathbb{R}^{3}$
as follows:
\begin{align}
	d_{i} & =||h_{i}-P||\nonumber \\
	& =\sqrt{(x_{i}-x_{p})^{2}+(y_{i}-y_{p})^{2}+(z_{i}-z_{p})^{2}}\label{eq:UWB_di}
\end{align}
The expression in \eqref{eq:UWB_di} can be squared as follows:
\begin{align}
	d_{i}^{2} & =||h_{i}||^{2}+||P||^{2}-2h_{i}^{\top}P\label{eq:UWB_di2}
\end{align}
As such, for $j\neq i$, in view of \eqref{eq:UWB_di} and \eqref{eq:UWB_di2},
one has
\begin{align}
	d_{j}^{2} & =||h_{j}||^{2}+||P||^{2}-2h_{j}^{\top}P\label{eq:UWB_dk2}
\end{align}
Hence, from \eqref{eq:UWB_di2} and \eqref{eq:UWB_dk2} where $i=1$
and $j=2$, one shows $d_{1}^{2}-d_{2}^{2}+||h_{2}||^{2}-||h_{1}||^{2}=2(h_{2}-h_{1})^{\top}P$.
Hence, one shows
\[
\underbrace{\left[\begin{array}{c}
		h_{2}^{\top}-h_{1}^{\top}\\
		h_{3}^{\top}-h_{1}^{\top}\\
		\vdots\\
		h_{N}^{\top}-h_{1}^{\top}
	\end{array}\right]}_{A}P=\underbrace{\frac{1}{2}\left[\begin{array}{c}
		d_{1}^{2}-d_{2}^{2}+||h_{2}||^{2}-||h_{1}||^{2}\\
		d_{1}^{2}-d_{3}^{2}+||h_{3}||^{2}-||h_{1}||^{2}\\
		\vdots\\
		d_{1}^{2}-d_{N}^{2}+||h_{N}||^{2}-||h_{1}||^{2}
	\end{array}\right]}_{B}
\]
where $N$ denotes the number of fixed anchors accessed by the tag.
Defining $\delta=\frac{1}{2}(AP-B)^{\top}(AP-B)$ and applying the
Minimum Mean Square Error (MMSE) method, one obtains $\frac{\partial\delta}{\partial P}=A^{\top}(AP-B)=0$
such that
\begin{equation}
	P=(A^{\top}A)^{-1}A^{\top}B\label{eq:UWB_P}
\end{equation}
\begin{assum}\label{Assumption:assum_TOA} To guarantee that the
	vehicle position $P$ is uniquely defined and $(A^{\top}A)^{-1}$
	is nonsingular, the tag must be within range of at least 4 anchors
	($N\geq4$) for positioning in 3D space and at least 3 anchors ($N\geq3$)
	for positioning in 2D space.\end{assum}

\subsection{UWB and Time Difference Of Arrival\label{subsec:TDOA}}

TOA-based range measurements rely on synchronization between the tag
and the anchor nodes. As such, the implementation of a TOA system
is rather complex and is rarely used in practice. Time Difference
Of Arrival (TDOA) approach, on the other hand, circumvents the need
for synchronization and thereby is a common choice \cite{hashim2023nonlinear,bottigliero2021low,sidorenko2020error}.
TDOA defines the ranging distance as the difference between arrival
time of a transmitted signal from two source anchors to the target
tag. Based on the TDOA approach, the range distance $d_{j,i}$ between
the UWB tag positioned at $P=[x,y,z]^{\top}\in\mathbb{R}^{3}$ (attached
to the vehicle) and the two anchors positioned at $h_{i}=[x_{i},y_{i},z_{i}]^{\top}\in\mathbb{R}^{3}$
and $h_{j}=[x_{j},y_{j},z_{j}]^{\top}\in\mathbb{R}^{3}$ is defined
as follows:
\begin{align}
	d_{j,i}= & ||P-h_{j}||-||P-h_{i}||\label{eq:UWB_dji}
\end{align}
The expression in \eqref{eq:UWB_dji} can be squared such that
\begin{align}
	\frac{d_{j,i}^{2}+||h_{i}||^{2}-||h_{j}||^{2}}{2}= & (h_{i}-h_{j})^{\top}P-d_{j,i}||P-h_{i}||\label{eq:UWB_dji2}
\end{align}
\begin{figure}[h]
	\centering{}\centering\includegraphics[scale=0.55]{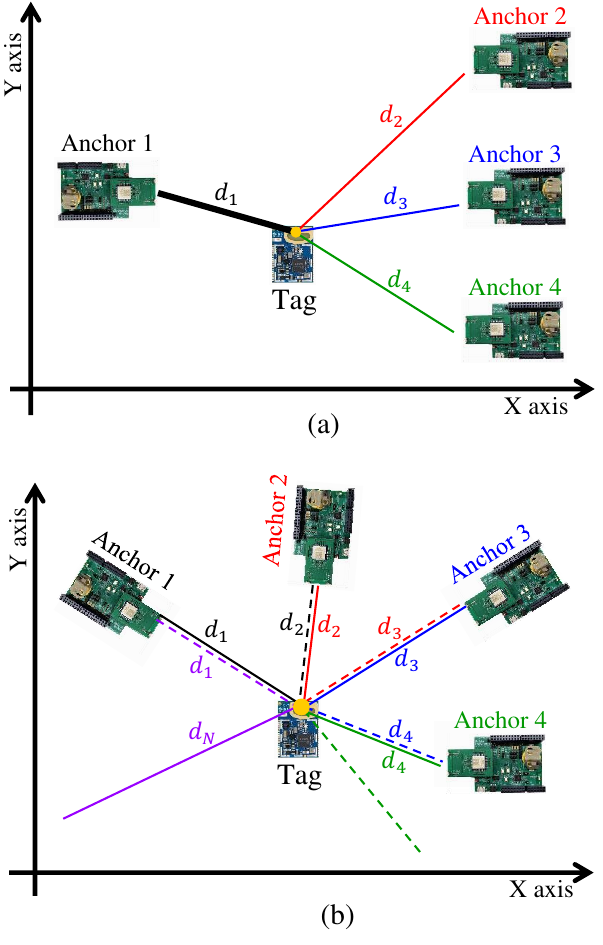}\caption{\label{fig:TDOA1} Topology of TDOA-based localization system.}
\end{figure}

One way to localize the tag involves using the range difference between
the main Base Station (BS) and other BSs. Fig. \ref{fig:TDOA1}.(a)
illustrates the topology of TDOA-based localization system composed
of a main BS (anchor 1), other BSs (anchor 2, 3, 4), and the tag.
From \eqref{eq:UWB_dji} and \eqref{eq:UWB_dji2}, considering $N$
TDOA measurements, one has\begin{small}
	\begin{align*}
		\underbrace{\left[\begin{array}{cc}
				(h_{1}-h_{2})^{\top} & -d_{2,1}\\
				(h_{1}-h_{3})^{\top} & -d_{3,1}\\
				\vdots & \vdots\\
				(h_{1}-h_{N})^{\top} & -d_{N,1}
			\end{array}\right]}_{A}\overline{P} & =\underbrace{\frac{1}{2}\left[\begin{array}{c}
				d_{2,1}^{2}+||h_{1}||^{2}-||h_{2}||^{2}\\
				d_{3,1}^{2}+||h_{1}||^{2}-||h_{3}||^{2}\\
				\vdots\\
				d_{N,1}^{2}+||h_{1}||^{2}-||h_{N}||^{2}
			\end{array}\right]}_{B}
	\end{align*}
\end{small}where $N$ represents the number of fixed anchors accessed
by the tag and $\overline{P}=[P^{\top},||P-h_{1}||]^{\top}\in\mathbb{R}^{4}$.
Defining $\delta=\frac{1}{2}(AP-B)^{\top}(AP-B)$ and applying MMSE,
one can show that $\frac{\partial\delta}{\partial P}=A^{\top}(AP-B)=0$
with
\begin{equation}
	\overline{P}=(A^{\top}A)^{-1}A^{\top}B\label{eq:UWB_Pbar}
\end{equation}
Alternatively, positioning can be achieved without utilizing a main
BS station as presented in Fig. \ref{fig:TDOA1}.(b). From the expressions
in \eqref{eq:UWB_dji} and \eqref{eq:UWB_dji2}, and for $N$ TDOA
measurements, one obtains\begin{small}
	\begin{align}
		\frac{d_{2,1}^{2}+||h_{1}||^{2}-||h_{2}||^{2}}{2}= & (h_{1}-h_{2})^{\top}P-d_{2,1}||P-h_{1}||\nonumber \\
		\frac{d_{3,2}^{2}+||h_{2}||^{2}-||h_{3}||^{2}}{2}= & (h_{2}-h_{3})^{\top}P-d_{3,2}||P-h_{2}||\nonumber \\
		\vdots\nonumber \\
		\frac{d_{1,N}^{2}+||h_{N}||^{2}-||h_{1}||^{2}}{2}= & (h_{N}-h_{1})^{\top}P-d_{1,N}||P-h_{N}||\label{eq:UWB_djin}
	\end{align}
\end{small}Since, $||P-h_{3}||=d_{3,2}+||P-h_{2}||$, one shows
\begin{align*}
	||P-h_{3}|| & =d_{3,2}+d_{2,1}+||P-h_{1}||
\end{align*}
Likewise, one finds
\begin{align*}
	||P-h_{4}|| & =d_{4,3}+d_{3,2}+d_{2,1}+||P-h_{1}||
\end{align*}
Therefore, for $N$ TDOA measurements
\begin{align*}
	||P-h_{N}|| & =\sum_{i=2}^{N}d_{i,i-1}+||P-h_{1}||
\end{align*}
Let us define:
\[
A=\left[\begin{array}{cc}
	(h_{1}-h_{2})^{\top} & -d_{2,1}\\
	(h_{2}-h_{3})^{\top} & -d_{3,2}\\
	\vdots & \vdots\\
	(h_{N-1}-h_{N})^{\top} & -d_{N,N-1}\\
	(h_{N}-h_{1})^{\top} & -d_{1,N}
\end{array}\right]
\]
\[
B=\frac{1}{2}\left[\begin{array}{c}
	d_{2,1}^{2}+||h_{1}||^{2}-||h_{2}||^{2}\\
	d_{3,2}^{2}+||h_{2}||^{2}-||h_{3}||^{2}+2d_{3,2}\sum_{i=2}^{2}d_{i,i-1}\\
	d_{4,3}^{2}+||h_{3}||^{2}-||h_{4}||^{2}+2d_{4,3}\sum_{i=2}^{3}d_{i,i-1}\\
	\vdots\\
	d_{1,N}^{2}+||h_{N}||^{2}-||h_{1}||^{2}+2d_{1,N}\sum_{i=2}^{N}d_{i,i-1}
\end{array}\right]
\]
with $N$ being the number of fixed anchors accessed by the tag and
$\overline{P}=[P^{\top},||P-h_{1}||]^{\top}\in\mathbb{R}^{4}$. Hence,
the position determination is as follows:
\begin{equation}
	\overline{P}=(A^{\top}A)^{-1}A^{\top}B\label{eq:UWB_Pbar1}
\end{equation}
Note that Assumption \ref{Assumption:assum_TOA} holds for the TDOA
approach.

\subsection{Inertial Measurement Unit\label{subsec:IMU}}

A typical IMU is composed of three sensors: a gyroscope, an accelerometer,
and a magnetometer which are sufficient for attitude determination
\cite{hashim2018SO3Stochastic,markley2006attitude,d2018uav}. The
gyroscope provides vehicle's angular velocity measurements in the
$\{\mathcal{B}\}$-frame:
\begin{equation}
	\Omega_{m}=\Omega+n_{\Omega}\in\mathbb{R}^{3}\label{eq:UWB_Angular}
\end{equation}
where $\Omega$ represents the true angular velocity and $n_{\Omega}$
stands for unknown noise in measurements, for all $\Omega_{m},n_{\Omega}\in\{\mathcal{B}\}$.
The accelerometer provides measurements of the apparent acceleration:
\begin{equation}
	a_{m}=R^{\top}(\dot{V}-\overrightarrow{\mathtt{g}})+n_{a}\in\mathbb{R}^{3}\label{eq:UWB_am}
\end{equation}
where $\overrightarrow{\mathtt{g}}=[0,0,g]^{\top}$ and $\dot{V}$
denote the gravitational acceleration and the linear acceleration,
respectively, with respect to $\{\mathcal{I}\}$-frame. $n_{a}$ describes
unknown measurement noise. Note that $\overrightarrow{\mathtt{g}}$
is described in the North-East-Down (NED) frame $\forall a_{m},n_{a}\in\{\mathcal{B}\}$.
At a low frequency, $|\overrightarrow{\mathtt{g}}|>>|\dot{V}|$ and
the accelerometer measurements can be described as $a_{m}\approx-R^{\top}\overrightarrow{\mathtt{g}}+n_{a}$.
The magnetometer measurements can be represented by
\begin{equation}
	m_{m}=R^{\top}m_{r}+n_{m}\in\mathbb{R}^{3}\label{eq:UWB_Mm}
\end{equation}
where $m_{r}=[m_{N},0,m_{D}]^{\top}\in\{\mathcal{I}\}$ denotes the
earth-magnetic field in the NED frame and $n_{m}$ represents the
additive unknown noise components, $\forall m_{m},\omega_{m}\in\{\mathcal{B}\}$.
For the sake of attitude orthogonality, normalization of IMU vector
measurements ($a_{m},m_{m}$) and observations ($\dot{V}-\overrightarrow{\mathtt{g}},m_{r}$)
is commonly employed as follows:
\begin{equation}
	\begin{cases}
		v_{1}=\frac{a_{m}}{||a_{m}||}, & r_{1}=\frac{-\overrightarrow{\mathtt{g}}}{||-\overrightarrow{\mathtt{g}}||}\\
		v_{2}=\frac{m_{m}}{||m_{m}||}, & r_{2}=\frac{m_{r}}{||m_{r}||}\\
		v_{3}=\frac{v_{1}\times v_{2}}{||v_{1}\times v_{2}||}, & r_{3}=\frac{r_{1}\times r_{2}}{||r_{1}\times r_{2}||}
	\end{cases}\label{eq:UWB_IMU_Normal}
\end{equation}
The expression in \eqref{eq:UWB_IMU_Normal} ensures availability
of 3 non-collinear measurements/observations necessary for attitude
estimation \cite{hashim2018SO3Stochastic,markley2006attitude,d2018uav}.

\subsection{Navigation Problem\label{subsec:Navigation}}

The true navigation kinematics of a vehicle traveling in 3D space
are defined by
\begin{equation}
	\begin{cases}
		\dot{R} & =R\left[\Omega\right]_{\times}\\
		\dot{P} & =V\\
		\dot{V} & =Ra+\overrightarrow{\mathtt{g}}
	\end{cases},\hspace{1em}\underbrace{\dot{X}=XU-\mathcal{\mathcal{G}}X}_{\text{Compact form}}\label{eq:UWB_NAV_dot}
\end{equation}
with $R\in\mathbb{SO}\left(3\right)$ being the true attitude, $P\in\mathbb{R}^{3}$
being the true position, $V\in\mathbb{R}^{3}$ being the true linear
velocity, $\Omega\in\mathbb{R}^{3}$ being the true angular velocity,
$\overrightarrow{\mathtt{g}}$ denoting the gravity vector, and $a\in\mathbb{R}^{3}$
being the apparent acceleration (all non-gravitational forces on the
vehicle) for all $R,\Omega,a\in\{\mathcal{B}\}$ and $P,V\in\{\mathcal{I}\}$.
The right portion of \eqref{eq:UWB_NAV_dot} represents the navigation
kinematics in a compact form with $X\in\mathbb{SE}_{2}\left(3\right)$
as per the map defined in \eqref{eq:NAV_u}, while $U=u([\Omega\text{\ensuremath{]_{\times}}},0_{3\times1},a,1)\in\mathcal{U}_{m}$,
and $\mathcal{\mathcal{G}}=u(0_{3\times3},0_{3\times1},-\overrightarrow{\mathtt{g}},1)\in\mathcal{U}_{m}$
as per the map defined in \eqref{eq:NAV_u}. For more information
see \cite{hashim2021_COMP_ENG_PRAC}. In view of the navigation model
in \eqref{eq:UWB_NAV_dot}, the measurements of $\Omega$ and $a$
are given by
\begin{equation}
	\begin{cases}
		\Omega_{m} & =\Omega+n_{\Omega}\in\mathbb{R}^{3}\\
		a_{m} & =a+n_{a}\in\mathbb{R}^{3}
	\end{cases}\label{eq:UWB_NAV_Noise}
\end{equation}
where $n_{\Omega}$ and $n_{a}$ refer to Gaussian unknown bounded
and zero-mean noise. Since, derivative of a Gaussian process lead
to a Gaussian process \cite{hashim2021_COMP_ENG_PRAC,tong2011observer,ito1984lectures}, the
noise can be redefined as $n_{\Omega}=\mathcal{Q}d\beta_{\Omega}/dt$
and $n_{a}=\mathcal{Q}d\beta_{a}/dt$, function of Brownian motion
process vectors \cite{hashim2018SO3Stochastic,jazwinski2007stochastic}
where $\mathcal{Q}={\rm diag}(\mathcal{Q}_{1,1},\mathcal{Q}_{2,2},\mathcal{Q}_{3,3})\in\mathbb{R}^{3\times3}$
refers to an unknown diagonal matrix (diagonal is positive and time-variant)
and ${\rm diag}(\cdot)$ refers to a diagonal of a matrix. Thus, the
covariance of $n_{\Omega}$ and $n_{a}$ is given by $\mathcal{Q}^{2}=\mathcal{Q}\mathcal{Q}^{\top}$
(for more details visit \cite{hashim2018SO3Stochastic}). In view
of \eqref{eq:NAV_Ecul_Dist}, \eqref{eq:UWB_NAV_dot}, and \eqref{eq:UWB_NAV_Noise},
the kinematics in \eqref{eq:UWB_NAV_dot} can be re-expressed to follow
stochastic differential equations as follows:
\begin{equation}
	\begin{cases}
		d||R||_{{\rm I}} & =(1/2)\mathbf{vex}(\boldsymbol{\mathcal{P}}_{a}(R))^{\top}(\Omega_{m}dt-\mathcal{Q}d\beta_{\Omega})\\
		dP & =Vdt\\
		dV & =(Ra_{m}+\overrightarrow{\mathtt{g}})dt-R\mathcal{Q}d\beta_{a}
	\end{cases}\label{eq:UWB_NAV_STCH_dot}
\end{equation}
with ${\rm Tr}\{R[\Omega_{m}]_{\times}\}=-2\mathbf{vex}(\boldsymbol{\mathcal{P}}_{a}(R))^{\top}\Omega_{m}$
as defined in \eqref{eq:UWB_Identity2}. This means that \eqref{eq:UWB_NAV_STCH_dot}
is described by
\begin{align}
	dx & =fdt+h\overline{\mathcal{Q}}d\beta\label{eq:UWB_NAV_STCH_dot1}
\end{align}
with $x=[||R||_{{\rm I}},P^{\top},V^{\top}]^{\top}\in\mathbb{R}^{7}$,
$f=[(1/2)\mathbf{vex}(\boldsymbol{\mathcal{P}}_{a}(R))^{\top}\Omega_{m},V^{\top},(Ra_{m}+\overrightarrow{\mathtt{g}})^{\top}]^{\top}\in\mathbb{R}^{7}$,
and $h\overline{\mathcal{Q}}d\beta=[(1/2)\mathbf{vex}(\boldsymbol{\mathcal{P}}_{a}(R))^{\top}d\beta_{\Omega}^{\top}\mathcal{Q},0_{3\times1}^{\top},d\beta_{a}^{\top}\mathcal{Q}]^{\top}\in\mathbb{R}^{7}$.
Let us define the following variable:
\begin{equation}
	\sigma=[\sup_{t\geq0}\mathcal{Q}_{1,1},\sup_{t\geq0}\mathcal{Q}_{2,2},\sup_{t\geq0}\mathcal{Q}_{3,3}]^{\top}\in\mathbb{R}^{3}\label{eq:NAV_s}
\end{equation}

\begin{defn}
	\label{def:NAV_SGUUB}\cite{hashim2018SO3Stochastic,ji2006adaptive}
	Recall the stochastic kinematics in \eqref{eq:UWB_NAV_STCH_dot1}.
	The state vector $x(t)$ is almost SGUUB if for initial state $x(t_{0})$
	and a known set $\varXi\in\mathbb{R}^{7}$ there is a positive constant
	$k_{c}$ and a time constant $t_{c}=t_{c}(x(t_{0}))$ with $\mathbb{E}[||x(t_{0})||]<k_{c},\forall t>t_{0}+k_{c}$.
\end{defn}
\begin{lem}
	\label{Lemm:NAV_deng}\cite{deng2001stabilization} Recall the stochastic
	differential system in \eqref{eq:UWB_NAV_STCH_dot1} and consider
	$\mathbb{U}(x)$ to be a twice differentiable cost function where
	$\mathbb{U}:\mathbb{R}^{7}\rightarrow\mathbb{R}_{+}$ such that
	\begin{equation}
		\mathcal{L}\mathbb{U}(x)=(\frac{\partial\mathbb{U}}{\partial x})^{\top}f+\frac{1}{2}{\rm Tr}\{h\overline{\mathcal{Q}}^{2}h^{\top}\frac{\partial^{2}\mathbb{U}}{\partial x^{2}}\}\label{eq:UWB_Vfunction_Lyap0}
	\end{equation}
	with $\mathcal{L}\mathbb{U}$ referring to a differential operator.
	Let us define $\alpha_{1}(\cdot)$ and $\alpha_{2}(\cdot)$ as class
	$\mathcal{K}_{\infty}$ functions and define the following constants
	$z_{1}>0$ and $z_{2}\geq0$ such that
	\begin{align}
		\alpha_{1}(x) & \leq\mathbb{U}(x)\leq\alpha_{2}(x)\label{eq:UWB_Vfunction_Lyap}\\
		\mathcal{L}\mathbb{U}(x) & \leq-z_{1}\mathbb{U}(x)+z_{2}\label{eq:UWB_dVfunction_Lyap}
	\end{align}
	Therefore, the stochastic kinematics in \eqref{eq:UWB_NAV_STCH_dot}
	have an almost unique strong solution on $[0,\infty)$ and the solution
	$x$ is bounded in probability where
	\begin{equation}
		\mathbb{E}[\mathbb{U}(x)]\leq\mathbb{U}(x(0)){\rm exp}(-z_{1}t)+z_{2}/z_{1}\label{eq:UWB_EVfunction_Lyap}
	\end{equation}
	ensuring that $x$ is SGUUB in the mean square.
\end{lem}
\begin{lem}
	\label{Lemm:UWB_Lemma2}\cite{hashim2019SO3Wiley} Consider $R\in\mathbb{SO}\left(3\right)$
	and $M_{r}=M_{r}^{\top}\in\mathbb{R}^{3\times3}$ and define $\overline{M_{r}}={\rm Tr}\{M_{r}\}\mathbf{I}_{3}-M_{r}$
	with $\overline{\lambda}_{\overline{M_{r}}}$ and $\underline{\lambda}_{\overline{M_{r}}}$
	referring to the minimum and maximum eigenvalues of $\overline{M_{r}}$,
	respectively. Let $||M_{r}R||_{{\rm I}}=\frac{1}{4}{\rm Tr}\{M_{r}(\mathbf{I}_{3}-R)\}$
	and $\boldsymbol{\Upsilon}(M_{r}R)=\mathbf{vex}(\boldsymbol{\mathcal{P}}_{a}(M_{r}R))$.
	Thus, one obtains
	\begin{align}
		||\boldsymbol{\Upsilon}(M_{r}R)||^{2} & \leq2\overline{\lambda}_{\overline{M_{r}}}||M_{r}R||_{{\rm I}}\label{eq:UWB_lemm2_1}\\
		||\boldsymbol{\Upsilon}(M_{r}R)||^{2} & \geq\frac{\underline{\lambda}_{\overline{M_{r}}}}{2}||M_{r}R||_{{\rm I}}(1+{\rm Tr}\{R\})\label{eq:UWB_lemm2_2}
	\end{align}
\end{lem}

\section{Stochastic Navigation Filter \label{sec:UWB_Filter}}

In this section, a novel nonlinear stochastic complementary filter
that operates based on the fusion of UWB and IMU measurements is proposed.
Let $\hat{R}\in\mathbb{SO}\left(3\right)$, $\hat{P}\in\mathbb{R}^{3}$,
and $\hat{V}\in\mathbb{R}^{3}$ denote the estimates of attitude,
position, and linear velocity, respectively. Define the errors between
the true and the estimated values of attitude ($\tilde{R}$), position
($\tilde{P}$), and linear velocity ($\tilde{V}$) as follows:
\begin{equation}
	\begin{cases}
		\tilde{R} & =R\hat{R}^{\top}\\
		\tilde{P} & =P-\hat{P}\\
		\tilde{V} & =V-\hat{V}
	\end{cases}\label{eq:UWB_NAV_error}
\end{equation}
Let $\hat{\sigma}$ be the upper bound covariance estimate of $\sigma$
and define the estimation error as
\begin{equation}
	\tilde{\sigma}=\sigma-\hat{\sigma}\in\mathbb{R}^{3}\label{eq:UWB_NAV_s_error}
\end{equation}
For attitude estimation, our objective is to utilize direct vector
observations and measurements in the implementation. As such, from
\eqref{eq:UWB_IMU_Normal}, define 
\begin{align}
	M_{r} & =\sum_{i=1}^{3}s_{i}r_{i}r_{i}^{\top},\hspace{1em}M_{B}=\sum_{i=1}^{3}s_{i}v_{i}v_{i}^{\top}\label{eq:UWB_Mr_Mv}
\end{align}
where $s_{i}$ denotes the $i$th sensor confidence such that $\sum_{i=1}^{3}s_{i}=3$.
Define
\begin{equation}
	\hat{v}_{i}=\hat{R}^{\top}r_{i},\hspace{1em}\forall i=1,2,3\label{eq:UWB_v_est}
\end{equation}
Therefore, one finds
\begin{align}
	\mathbf{vex}(\boldsymbol{\mathcal{P}}_{a}(M_{r}\tilde{R})) & =\frac{1}{2}\mathbf{vex}(M_{r}\tilde{R}-\tilde{R}^{\top}M_{r})\nonumber \\
	& =\frac{1}{2}\mathbf{vex}\left(\sum_{i=1}^{3}s_{i}r_{i}v_{i}^{\top}\hat{R}^{\top}-\sum_{i=1}^{3}s_{i}\hat{R}v_{i}r_{i}^{\top}\right)\nonumber \\
	& =\frac{1}{2}\sum_{i=1}^{3}\hat{R}s_{i}(v_{i}\times\hat{v}_{i})\label{eq:UWB_VEX}
\end{align}
where $[v_{i}\times\hat{v}_{i}]_{\times}=\hat{v}_{i}v_{i}^{\top}-v_{i}\hat{v}_{i}^{\top}$.
Also, one shows
\begin{align}
	E_{r}=||M_{r}\tilde{R}||_{{\rm I}} & =\frac{1}{4}{\rm Tr}\{M_{r}(\mathbf{I}_{3}-\tilde{R})\}\nonumber \\
	& =\frac{1}{4}{\rm Tr}\left\{ M_{r}-\sum_{i=1}^{3}s_{i}\hat{R}\hat{v}_{i}v_{i}^{\top}\hat{R}^{\top}\right\} \label{eq:UWB_Er}
\end{align}
where $E_{r}:\mathbb{SO}\left(3\right)\rightarrow\mathbb{R}_{+}$
with $E_{r}>0\forall\tilde{R}\neq\mathbf{I}_{3}$ and $E_{r}=0$ at
$\tilde{R}=\mathbf{I}_{3}$. To this end, let $P_{y}$ denote a reconstructed
position satisfying Assumption \ref{Assumption:assum_TOA} which can
be obtained, for instance, as follows: 
\begin{equation}
	\begin{cases}
		P_{y}=(A^{\top}A)^{-1}A^{\top}B, & \text{TOA}\\
		\overline{P}_{y}=\left[\begin{array}{c}
			P_{y}\\
			||P_{y}-h_{1}||
		\end{array}\right]=(A^{\top}A)^{-1}A^{\top}B, & \text{TDOA}
	\end{cases}\label{eq:UWB_Py}
\end{equation}
where $A$ and $B$ matrices are defined in Section \ref{sec:SE3_Problem-Formulation}.
Let us define the following set of equations which includes a covariance
adaptation mechanism and correction factors:
\begin{equation}
	\begin{cases}
		E_{r} & =\frac{1}{4}{\rm Tr}\sum_{i=1}^{3}s_{i}(r_{i}r_{i}^{\top}-\hat{R}\hat{v}_{i}v_{i}^{\top}\hat{R}^{\top})\\
		\mathcal{D}_{v} & ={\rm diag}(\sum_{i=1}^{3}s_{i}v_{i}\times\hat{v}_{i})\\
		\dot{\hat{\sigma}} & =\gamma_{\sigma}\frac{E_{r}+2}{8}\exp(E_{r})\mathcal{D}_{{\rm v}}(\sum_{i=1}^{n}s_{i}v_{i}\times\hat{v}_{i})-k_{\sigma}\gamma_{\sigma}\hat{\sigma}\\
		w_{\Omega} & =-\frac{k_{1}}{2}\sum_{i=1}^{n}\hat{R}s_{i}(v_{i}\times\hat{v}_{i})-\frac{1}{8}\frac{E_{r}+2}{E_{r}+1}\hat{R}\mathcal{D}_{{\rm v}}\hat{\sigma}\\
		w_{V} & =-\frac{k_{v}}{\varepsilon}(P_{y}-\hat{P})-[w_{\Omega}]_{\times}\hat{P}\\
		w_{a} & =-k_{a}(P_{y}-\hat{P})-[w_{\Omega}]_{\times}\hat{V}
	\end{cases}\label{eq:UWB_Filter1_Correc}
\end{equation}
where $\gamma_{\sigma}$, $k_{1}$, $k_{v}$, $k_{a}$, $\varepsilon$,
and $k_{\sigma}$ are positive constants. Now let us propose the following
navigation stochastic filter design:
\begin{equation}
	\begin{cases}
		\dot{\hat{R}} & =\hat{R}\left[\Omega_{m}\right]_{\times}-\left[w_{\Omega}\right]_{\times}\hat{R}\\
		\dot{\hat{P}} & =\hat{V}-\left[w_{\Omega}\right]_{\times}\hat{P}-w_{V}\\
		\dot{\hat{V}} & =\hat{R}a_{m}+\overrightarrow{\mathtt{g}}-\left[w_{\Omega}\right]_{\times}\hat{V}-w_{a}
	\end{cases},\hspace{1em}\underbrace{\dot{\hat{X}}=\hat{X}U_{m}-W\hat{X}}_{\text{Compact form}}\label{eq:UWB_Filter1_Detailed}
\end{equation}
Quaternion form of the stochastic filter design proposed above is
outlined in the \nameref{sec:UWB_AppendixA}. For the compact form,
$\hat{X}\in\mathbb{SE}_{2}\left(3\right)$ describes the estimate
of $X$, $U_{m}=u([\Omega_{m}\text{\ensuremath{]_{\times}}},0_{3\times1},a_{m},1)\in\mathcal{U}_{\mathcal{M}}$,
and $W=u([w_{\Omega}\text{\ensuremath{]_{\times}}},w_{V},w_{a},1)\in\mathcal{U}_{\mathcal{M}}$,
refer to the map in \eqref{eq:NAV_u}. For simplicity's sake, in the
analysis, $P=P_{y}$.
\begin{thm}
	\label{thm:Theorem1} Recall the nonlinear stochastic differential
	system in \eqref{eq:UWB_NAV_STCH_dot}. Consider that at each time
	instant, at least 3 non-collinear measurements/observations as to
	\eqref{eq:UWB_IMU_Normal} are available also consider Assumption
	\ref{Assumption:assum_TOA} holds true (for any of TOA- or TDOA-based
	approaches). Let the nonlinear navigation stochastic differential
	estimator in \eqref{eq:UWB_Filter1_Detailed} be integrated with the
	direct measurements and innovation terms in \eqref{eq:UWB_Py} and
	\eqref{eq:UWB_Filter1_Correc} such that $\Omega_{m}=\Omega+n_{\Omega}$
	and $a_{m}=a+n_{a}$. Thus, all the closed-loop signals are almost
	semi-globally uniformly ultimately bounded in the mean square.
\end{thm}
\begin{proof}From \eqref{eq:UWB_NAV_STCH_dot}, \eqref{eq:UWB_Er},
	\eqref{eq:UWB_NAV_error}, and \eqref{eq:UWB_Filter1_Detailed}, one
	obtains
	\begin{align}
		\frac{d}{dt}E_{r}= & \frac{d}{dt}\frac{1}{4}{\rm Tr}\{M_{r}(\mathbf{I}_{3}-\tilde{R})\}\nonumber \\
		= & -\frac{1}{4}{\rm Tr}\{M_{r}\tilde{R}[w_{\Omega}]_{\times}\}+\frac{1}{4}{\rm Tr}\{M_{r}\tilde{R}[\hat{R}\mathcal{Q}_{\Omega}d\beta_{\Omega}]_{\times}\}\nonumber \\
		= & \frac{1}{2}\mathbf{vex}(\boldsymbol{\mathcal{P}}_{a}(M_{r}\tilde{R}))^{\top}(w_{\Omega}dt-\hat{R}\mathcal{Q}_{\Omega}d\beta_{\Omega})\label{eq:UWB_Er_dot}
	\end{align}
	Note that $M_{r}$ is a constant matrix and ${\rm Tr}\{M_{r}\tilde{R}[w_{\Omega}]_{\times}\}={\rm Tr}\{\boldsymbol{\mathcal{P}}_{a}(M_{r}\tilde{R})[w_{\Omega}]_{\times}\}=-\frac{1}{2}\mathbf{vex}(\boldsymbol{\mathcal{P}}_{a}(M_{r}\tilde{R}))^{\top}w_{\Omega}$.
	From \eqref{eq:UWB_NAV_STCH_dot}, \eqref{eq:UWB_NAV_error}, and
	\eqref{eq:UWB_Filter1_Detailed}, one shows
	\begin{equation}
		\begin{cases}
			\dot{\tilde{P}} & =\tilde{V}+[w_{\Omega}]_{\times}\hat{P}+w_{V}\\
			d\tilde{V} & =((\tilde{R}-\mathbf{I}_{3})\hat{R}a+[w_{\Omega}]_{\times}\hat{V}+w_{a})dt-R\mathcal{Q}d\beta_{a}
		\end{cases}\label{eq:UWB_Filter1_Error_dot}
	\end{equation}
	Define $\mathbb{U}_{T}=\mathbb{U}_{T}(E_{r},\tilde{P},\tilde{V},\tilde{\sigma})$
	as a Lyapunov function candidate such that
	\begin{equation}
		\mathbb{U}_{T}=\mathbb{U}_{R}+\mathbb{U}_{PV}\label{eq:UWB_LyapT}
	\end{equation}
	where $\mathbb{U}_{T}:\mathbb{SO}\left(3\right)\times\mathbb{R}^{3}\times\mathbb{R}^{3}\times\mathbb{R}^{3}\rightarrow\mathbb{R}_{+}$.
	Now consider the following Lyapunov function candidate $\mathbb{U}_{R}:\mathbb{SO}\left(3\right)\times\mathbb{R}^{3}\rightarrow\mathbb{R}_{+}$:
	\begin{equation}
		\mathbb{U}_{R}=\exp(E_{r})E_{r}+\frac{1}{2\gamma_{\sigma}}||\tilde{\sigma}||^{2}\label{eq:UWB_LyapR}
	\end{equation}
	From \eqref{eq:UWB_Vfunction_Lyap0}, one can show that $\frac{\partial}{\partial E_{r}}\mathbb{U}_{R}=\exp(E_{r})(E_{r}+1)$
	and $\frac{\partial^{2}}{\partial E_{r}^{2}}\mathbb{U}_{R}=\exp(E_{r})(E_{r}+2)$.
	From \eqref{eq:UWB_dVfunction_Lyap} and \eqref{eq:UWB_Filter1_Detailed},
	one obtains
	\begin{align}
		\mathcal{L}\mathbb{U}_{R}= & \left(\frac{\partial\mathbb{U}_{R}}{\partial E_{r}}\right)^{\top}f_{R}+\frac{1}{2}{\rm Tr}\left\{ g_{R}\mathcal{Q}_{\Omega}^{2}g_{R}^{\top}\frac{\partial^{2}\mathbb{U}_{R}}{\partial e_{R}^{2}}\right\} -\frac{1}{\gamma_{\sigma}}\tilde{\sigma}^{\top}\dot{\hat{\sigma}}\nonumber \\
		= & \frac{1}{8}\frac{\partial^{2}U_{R}}{\partial E_{r}^{2}}\mathbf{vex}(\boldsymbol{\mathcal{P}}_{a}(M_{r}\tilde{R}))^{\top}\hat{R}\mathcal{Q}_{\Omega}^{2}\hat{R}^{\top}\mathbf{vex}(\boldsymbol{\mathcal{P}}_{a}(M_{r}\tilde{R}))\nonumber \\
		& +\frac{1}{2}\frac{\partial^{2}U_{R}}{\partial E_{r}^{2}}\mathbf{vex}(\boldsymbol{\mathcal{P}}_{a}(M_{r}\tilde{R}))^{\top}w_{\Omega}-\frac{1}{\gamma_{\sigma}}\tilde{\sigma}^{\top}\dot{\hat{\sigma}}\label{eq:UWB_LyapR1dot}
	\end{align}
	where $||\mathcal{Q}^{2}||_{F}\leq||{\rm diag}(\sigma)||_{F}$ as
	in \eqref{eq:NAV_s}. Substituting $\dot{\hat{\sigma}}$ and $w_{\Omega}$
	with their definitions in \eqref{eq:UWB_Filter1_Correc} and in view
	of \eqref{eq:UWB_lemm2_1} in Lemma \ref{Lemm:UWB_Lemma2}, it becomes
	apparent that
	\begin{align}
		\mathcal{L}\mathbb{U}_{R} & \leq-(1+{\rm Tr}\{\tilde{R}\})\frac{k_{1}\underline{\lambda}_{\overline{M_{r}}}}{4}\exp(E_{r})E_{r}+k_{\sigma}\tilde{\sigma}^{\top}\hat{\sigma}\nonumber \\
		& \leq-k_{1}c_{R}||M_{r}\tilde{R}||_{{\rm I}}-\frac{k_{\sigma}}{2}||\tilde{\sigma}||^{2}+\frac{k_{\sigma}}{2}||\sigma||^{2}\label{eq:UWB_LyapR_1dot_2}
	\end{align}
	with $k_{\sigma}\tilde{\sigma}^{\top}\sigma\leq\frac{k_{\sigma}}{2}||\sigma||^{2}+\frac{k_{\sigma}}{2}||\tilde{\sigma}||^{2}$
	(see Young's inequality) and $c_{R}=\frac{1}{4}\underline{\lambda}_{\overline{M_{r}}}(1+{\rm Tr}\{\tilde{R}\})$.
	Define $\underline{\lambda}_{R}=\min\{k_{1}c_{R},\frac{k_{\sigma}}{2}\}$
	and $e_{R}=[||M_{r}\tilde{R}||_{{\rm I}},||\tilde{\sigma}||]^{\top}$.
	Thus, $\mathcal{L}\mathbb{U}_{R}$ can be described as
	\begin{align}
		\mathcal{L}\mathbb{U}_{R} & \leq-\underline{\lambda}_{R}||e_{R}||^{2}+\frac{k_{\sigma}}{2}||\sigma||^{2}\label{eq:UWB_LyapR_1dot_3}
	\end{align}
	Consider the following real-valued function $\mathbb{U}_{PV}:\mathbb{R}^{3}\times\mathbb{R}^{3}\rightarrow\mathbb{R}_{+}$:
	\begin{align}
		\mathbb{U}_{PV} & =\frac{1}{4}||\tilde{P}||^{4}+\frac{1}{4k_{d}}||\tilde{V}||^{4}-\frac{1}{3\mu}||\tilde{V}||^{2}\tilde{V}^{\top}\tilde{P}\label{eq:UWB_LyapPV}
	\end{align}
	with $k_{d}$ and $\mu$ being positive constants. Hence, one obtains
	\begin{align}
		\frac{\partial\mathbb{U}_{PV}}{\partial\tilde{P}} & =||\tilde{P}||^{2}\tilde{P}-\frac{1}{3\mu}||\tilde{V}||^{2}\tilde{V}\nonumber \\
		\frac{\partial\mathbb{U}_{PV}}{\partial\tilde{V}} & =\frac{1}{k_{d}}||\tilde{V}||^{2}\tilde{V}-\frac{1}{\mu}||\tilde{V}||^{2}\tilde{P}\nonumber \\
		\frac{\partial^{2}\mathbb{U}_{PV}}{\partial\tilde{V}^{2}} & =\frac{1}{k_{d}}(||\tilde{V}||^{2}\mathbf{I}_{3}+2\tilde{V}\tilde{V}^{\top}-2\tilde{P}\tilde{V}^{\top})\label{eq:UWB_LyapPV_vv}
	\end{align}
	Recall that ${\rm Tr}\{\hat{R}\mathcal{Q}^{2}\hat{R}^{\top}\}={\rm Tr}\{\mathcal{Q}^{2}\}$,
	$\tilde{V}^{\top}[w_{\Omega}]_{\times}\tilde{V}=0$, and $||\mathbf{I}_{3}-\tilde{R}||_{F}=2\sqrt{2}\sqrt{||\tilde{R}||_{{\rm I}}}\leq4\overline{\lambda}_{M}\sqrt{||M\tilde{R}||_{{\rm I}}}$
	(see \cite{hashim2019SO3Wiley}). As such, in view of \eqref{eq:UWB_dVfunction_Lyap},
	\eqref{eq:UWB_Filter1_Error_dot}, \eqref{eq:UWB_LyapPV}, and \eqref{eq:UWB_LyapPV_vv},
	one finds
	\begin{align}
		& \mathcal{L}\mathbb{U}_{PV}=(||\tilde{P}||^{2}\tilde{P}-\frac{1}{3\mu}||\tilde{V}||^{2}\tilde{V})^{\top}(\tilde{V}+[w_{\Omega}]_{\times}\hat{P}+w_{V})\nonumber \\
		& \hspace{1.5em}+||\tilde{V}||^{2}(\frac{1}{k_{d}}\tilde{V}-\frac{1}{\mu}\tilde{P})^{\top}((\tilde{R}-\mathbf{I}_{3})\hat{R}a+[w_{\Omega}]_{\times}\hat{V}+w_{a})\nonumber \\
		& \hspace{1.5em}+\frac{1}{2k_{d}}{\rm Tr}\left\{ (||\tilde{V}||^{2}\mathbf{I}_{3}+2\tilde{V}\tilde{V}^{\top}-2\tilde{P}\tilde{V}^{\top})R\mathcal{Q}^{2}R^{\top}\right\} \nonumber \\
		& \hspace{0.5em}\leq-c_{1}||\tilde{P}||^{4}-c_{2}||\tilde{V}||^{4}+c_{3}||\tilde{V}||^{2}||\tilde{P}||^{2}\nonumber \\
		& \hspace{1.5em}+2c_{g}||\tilde{V}||^{2}\,\sqrt{||M\tilde{R}||_{{\rm I}}}+\frac{1}{4k_{d}}||\sigma||^{2}\label{eq:UWB_LyapPV_1dot_1}
	\end{align}
	where $c_{m}=\max\{3\mu k_{d},3\mu,6k_{d}\mu\overline{\lambda}_{M}c_{a},6k_{d}k_{a}\}$,
	$c_{a}=\sup_{t\geq0}||a||$, $c_{1}=\frac{2k_{v}-\varepsilon}{2\varepsilon}$,
	$c_{2}=\frac{k_{d}-4\mu}{3\mu k_{d}}$, $c_{3}=(\varepsilon c_{m}+k_{v}k_{d})/\varepsilon k_{d}\mu$,
	and $c_{g}=\max\{\frac{4\overline{\lambda}_{M}||a||}{k_{d}},\frac{4\overline{\lambda}_{M}||a||}{\mu}\}$.
	Therefore, $\mathcal{L}\mathbb{U}_{PV}$ in \eqref{eq:UWB_LyapPV_1dot_1}
	becomes
	\begin{align}
		\mathcal{L}\mathbb{U}_{PV}\leq & -e_{PV}^{\top}\underbrace{\left[\begin{array}{cc}
				c_{1} & -\frac{c_{3}}{2}\\
				-\frac{c_{3}}{2} & c_{2}
			\end{array}\right]}_{Q_{PV}}e_{PV}\nonumber \\
		& +2c_{g}||e_{PV}||\,||e_{R}||+\frac{1}{4k_{d}}||\sigma||^{2}\label{eq:UWB_LyapPV_1dot_Final}
	\end{align}
	where $e_{PV}=[||\tilde{P}||^{2},||\tilde{V}||^{2}]^{\top}$. $Q_{PV}$
	is made positive by selecting $k_{v}>\varepsilon/2$, $k_{d}>4\mu$,
	and $c_{1}c_{2}>\frac{c_{3}^{2}}{4}$. Let us define $\underline{\lambda}_{PV}=\underline{\lambda}(Q_{PV})$.
	From \eqref{eq:UWB_LyapT}, \eqref{eq:UWB_LyapR_1dot_3}, and \eqref{eq:UWB_LyapPV_1dot_Final},
	the differential operator $\mathcal{L}\mathbb{U}_{R}$ and $\mathcal{L}\mathbb{U}_{PV}$
	can be combined and $\mathcal{L}\mathbb{U}_{T}$ can be expressed
	as
	\begin{align*}
		\mathcal{L}\mathbb{U}_{T}\leq & -\underline{\lambda}_{R}||e_{R}||^{2}-\underline{\lambda}_{PV}||e_{PV}||^{2}+2c_{g}||e_{PV}||\,||e_{R}||\\
		& +(\frac{1}{4k_{d}}+\frac{k_{\sigma}}{2})||\sigma||^{2}
	\end{align*}
	\begin{align}
		\mathcal{L}\mathbb{U}_{T}\leq & -e_{T}^{\top}\underbrace{\left[\begin{array}{cc}
				\underline{\lambda}_{R} & -c_{g}\\
				-c_{g} & \underline{\lambda}_{PV}
			\end{array}\right]}_{Q_{T}}e_{T}+\eta_{\sigma}\label{eq:UWB_LyapT_1dot_2}
	\end{align}
	where $\eta_{\sigma}=(\frac{1}{4k_{d}}+\frac{k_{\sigma}}{2})||\sigma||^{2}$
	and $e_{T}=[||e_{R}||,||e_{PV}||]^{\top}$. Hence, $Q_{T}$ is made
	positive by selecting $\underline{\lambda}_{R}>\frac{c_{g}^{2}}{4\underline{\lambda}_{PV}}$.
	Let us define $\underline{\lambda}_{T}=\underline{\lambda}(Q_{T})$.
	Hence, from \eqref{eq:UWB_LyapT_1dot_2}, one shows
	\begin{align}
		\mathcal{L}\mathbb{U}_{T}\leq & -\underline{\lambda}_{T}||e_{T}||^{2}+\eta_{\sigma}\label{eq:UWB_LyapT_1dot_Final}
	\end{align}
	In other words
	\begin{equation}
		d\mathbb{E}[\mathbb{U}_{T}]/dt=\mathbb{E}[\mathcal{L}\mathbb{U}_{T}]\leq-\underline{\lambda}_{T}\mathbb{E}[\mathbb{U}_{T}]+\eta_{\sigma}\label{eq:UWB_Lyap2_final3}
	\end{equation}
	In view of Lemma \ref{Lemm:NAV_deng}, it is obvious that
	\[
	0\leq\mathbb{E}[\mathbb{U}_{T}(t)]\leq\mathbb{U}_{T}(0){\rm exp}(-\underline{\lambda}_{T}t)+\eta_{\sigma}/\underline{\lambda}_{T}
	\]
	Thereby, it becomes apparent that $e_{T}$ or ($||M\tilde{R}||_{{\rm I}},||\tilde{P}||,||\tilde{V}||,||\tilde{\sigma}$\textbar\textbar )
	is almost SGUUB completing the proof.\end{proof}

\subsection{Discrete Implementation}

Define $\Delta t$ as a small sample time. Algorithm \ref{alg:Alg_Disc0}
details the implementation steps of the proposed continuous nonlinear
stochastic navigation filter \eqref{eq:UWB_Py}-\eqref{eq:UWB_Filter1_Detailed}
but in a discrete form. $\exp(\cdot)$ in Algorithm \ref{alg:Alg_Disc0}
refers to an exponential of a matrix, commonly termed ``expm''.
In the algorithm, the subscript $k$ denotes the $k$th iteration.
\begin{algorithm}[h]
	\caption{\label{alg:Alg_Disc0}Discrete navigation stochastic filter}
	
	\textbf{Initialization}:
	\begin{enumerate}
		\item[{\footnotesize{}1:}] Set $\hat{P}_{0|0},\hat{V}_{0|0},\hat{\sigma}_{0}\in\mathbb{R}^{3}$
		and $\hat{R}_{0|0}\in\mathbb{SO}\left(3\right)$.
		\item[{\footnotesize{}2:}] Set $k=0$ and pick the filter design parameters.
	\end{enumerate}
	\textbf{while }(1)\textbf{ do}
	\begin{enumerate}
		\item[] /{*} Prediction {*}/
		\item[{\footnotesize{}3:}] $\hat{X}_{k|k}=\left[\begin{array}{ccc}
			\hat{R}_{k|k} & \hat{P}_{k|k} & \hat{V}_{k|k}\\
			0_{1\times3} & 1 & 0\\
			0_{1\times3} & 0 & 1
		\end{array}\right]\in\mathbb{SE}_{2}\left(3\right)$ and \\
		$\hat{U}_{k}=\left[\begin{array}{ccc}
			[\Omega_{m}[k]\text{\ensuremath{]_{\times}}} & 0_{3\times1} & a_{m}[k]\\
			0_{1\times3} & 0 & 0\\
			0_{1\times3} & 1 & 0
		\end{array}\right]\in\mathcal{U}_{\mathcal{M}}$
		\item[{\footnotesize{}4:}] $\hat{X}_{k+1|k}=\hat{X}_{k|k}\exp(\hat{U}_{k}\Delta t)$
		\item[] /{*} Update step {*}/
		\item[{\footnotesize{}5:}] $\begin{cases}
			P_{y}=(A^{\top}A)^{-1}A^{\top}B, & \text{TOA}\\
			\overline{P}_{y}=\left[\begin{array}{c}
				P_{y}\\
				||P-h_{1}||
			\end{array}\right]=(A^{\top}A)^{-1}A^{\top}B, & \text{TDOA}
		\end{cases}$
		\item[{\footnotesize{}6:}] $\begin{cases}
			v_{1}=\frac{a_{m}}{||a_{m}||}, & r_{1}=\frac{-\overrightarrow{\mathtt{g}}}{||-\overrightarrow{\mathtt{g}}||}\\
			v_{2}=\frac{m_{m}}{||m_{m}||}, & r_{2}=\frac{m_{r}}{||m_{r}||}\\
			v_{3}=\frac{v_{1}\times v_{2}}{||v_{1}\times v_{2}||}, & r_{3}=\frac{r_{1}\times r_{2}}{||r_{1}\times r_{2}||}
		\end{cases}$
		\item[{\footnotesize{}7:}] $\begin{cases}
			E_{r} & =\frac{1}{4}{\rm Tr}\sum_{i=1}^{3}s_{i}(r_{i}r_{i}^{\top}-\hat{R}_{k|k}\hat{v}_{i}v_{i}^{\top}\hat{R}_{k|k}^{\top})\\
			\mathcal{D}_{v} & ={\rm diag}(\sum_{i=1}^{3}s_{i}v_{i}\times\hat{v}_{i})
		\end{cases}$
		\item[{\footnotesize{}8:}] $\begin{cases}
			\hat{\sigma}_{k+1} & =\hat{\sigma}_{k}-\Delta tk_{\sigma}\gamma_{\sigma}\hat{\sigma}_{k}\\
			& \hspace{1em}+\Delta t\gamma_{\sigma}\frac{E_{r}+2}{8}\exp(E_{r})\mathcal{D}_{{\rm v}}(\sum_{i=1}^{n}s_{i}v_{i}\times\hat{v}_{i})\\
			w_{\Omega} & =-\frac{k_{1}}{2}\sum_{i=1}^{n}\hat{R}s_{i}(v_{i}\times\hat{v}_{i})-\frac{1}{8}\frac{E_{r}+2}{E_{r}+1}\hat{R}\mathcal{D}_{{\rm v}}\hat{\sigma}_{k}\\
			w_{V} & =-\frac{k_{v}}{\varepsilon}(P_{y}-\hat{P}_{k|k})-[w_{\Omega}]_{\times}\hat{P}_{k|k}\\
			w_{a} & =-\overrightarrow{\mathtt{g}}-k_{a}(P_{y}-\hat{P}_{k|k})-[w_{\Omega}]_{\times}\hat{V}_{k|k}
		\end{cases}$
		\item[{\footnotesize{}9:}] $W_{k}=\left[\begin{array}{ccc}
			[w_{\Omega}[k]\text{\ensuremath{]_{\times}}} & w_{V}[k] & w_{a}[k]\\
			0_{1\times3} & 0 & 0\\
			0_{1\times3} & 1 & 0
		\end{array}\right]$
		\item[{\footnotesize{}10:}] $\hat{X}_{k+1|k+1}=\exp(-W_{k}\Delta t)\hat{X}_{k+1|k}$
		\item[{\footnotesize{}11:}] $k=k+1$
	\end{enumerate}
	\textbf{end while}
\end{algorithm}

\section{Pose determination from UWB\label{sec:UWB_Pose}}

The UWB industry is challenged with improving ranging accuracy. Improved
UWB ranging accuracy will constitute a breakthrough enabling pose
(orientation + position) determination based solely on UWB anchors
and tags. This will alleviate the need for a typical 9-axis IMU (gyroscope
+ accelerometer + magnetometer). Many vehicles (e.g., mobile robots
and drones) are equipped with a 6-axis IMU (gyroscope + accelerometer).
Thereby, obtaining an additional observation and measurement for orientation
determination requires equipping the vehicle with an additional motion
sensor such as a magnetometer to improve the attitude estimation accuracy.
Another approach suggests equipping the vehicle with a vision unit
(e.g., monocular or stereo camera) to enhance orientation and positioning
accuracy \cite{hashim2021_COMP_ENG_PRAC,beauvisage2021robust,cheng2014seamless}.
However, adding a vision unit increases the computational cost (e.g.,
feature extraction and pose estimation) and places additional load
on the battery power consumption. Considering the advantages offered
by the UWB technology, such as robustness against multipath interference
and fading \cite{bottigliero2021low,sidorenko2020error}, localization
in LOS and NLOS \cite{jiang2020uwb}, and low cost and power requirements
\cite{wen2020new,strohmeier2018ultra}, improving UWB ranging accuracy
will significantly advance the area of autonomous mobility and intelligent
transportation. Orientation and position of a vehicle equipped with
IMU and high accuracy UWB navigating in 3D space can be fully defined
if one of the following conditions is satisfied:
\begin{itemize}
	\item[(i)]  9-axis IMU, UWB anchors satisfying Assumption \ref{Assumption:assum_TOA},
	and a vehicle equipped with at least one tag.
	\item[(ii)]  6-axis IMU, UWB anchors satisfying Assumption \ref{Assumption:assum_TOA},
	and a vehicle equipped with at least two tags .
	\item[(iii)]  Gyroscope, UWB anchors satisfying Assumption \ref{Assumption:assum_TOA},
	and a vehicle equipped with at least three tags .
\end{itemize}
The problem formulation in Section \ref{sec:SE3_Problem-Formulation}
and the filter design in \ref{sec:UWB_Filter} consider that UWB anchors
satisfy Assumption \ref{Assumption:assum_TOA} such that the UWB tag
enables vehicle positioning and a 9-axis IMU enables orientation determination
(as detailed in Theorem \ref{thm:Theorem1}) which shows (i). Fig.
\ref{fig:UWB_Orient} illustrates a system with two tags placed at
a considerable distance and a 6-axis IMU capable of full pose determination.
Let $P=[x,y,z]^{\top}$ denote the vehicle's position (at the center
point) which also represents the position between $\{\mathcal{B}\}$-frame
and $\{\mathcal{I}\}$-frame. Assume that two UWB tags are placed
at the vehicle's edges as illustrated in Fig. \ref{fig:UWB_Orient}
$G_{1}=[x_{g1},y_{g1},z_{g1}]^{\top}$ and $G_{2}=[x_{g2},y_{g2},z_{g2}]^{\top}$
are described with respect to the $\{\mathcal{I}\}$-frame such that
$P=\frac{1}{2}(G_{1}+G_{2})$. In view of Assumption \ref{Assumption:assum_TOA},
the position of each tag can be easily defined. The transformation
from $G_{1}$ to $P$ or from $G_{2}$ to $P$ can be obtained using
vehicle's orientation $R$ and a coordination vector $s_{1}=[x_{s1},y_{s1},z_{s1}]^{\top}$.
$s_{1}$ is a known constant vector, and can be initially calculated
given known orientation (e.g., positioning the vehicle $\{\mathcal{B}\}$-frame
at the $\{\mathcal{I}\}$-frame). As such, one obtains $G_{1}=P+Rs_{1}$
or $G_{2}=P-Rs_{1}$. Given 6-axis IMU and two tags, the necessary
3 observations and measurements can be obtained as follows:
\begin{equation}
	\begin{cases}
		a_{m} & \approx-R^{\top}\overrightarrow{\mathtt{g}},\hspace{1em}\text{accelerometer}\\
		s_{1} & =R^{\top}(G_{1}-P)\\
		s_{2} & =\frac{a_{m}\times s_{1}}{||a_{m}\times s_{1}||},\hspace{1em}r_{3}=\frac{-\overrightarrow{\mathtt{g}}\times(G_{1}-P)}{||-\overrightarrow{\mathtt{g}}\times(G_{1}-P)||}
	\end{cases}\label{eq:UWB_IMU_UWB}
\end{equation}
\begin{figure}[h]
	\centering{}\centering\includegraphics[scale=0.75]{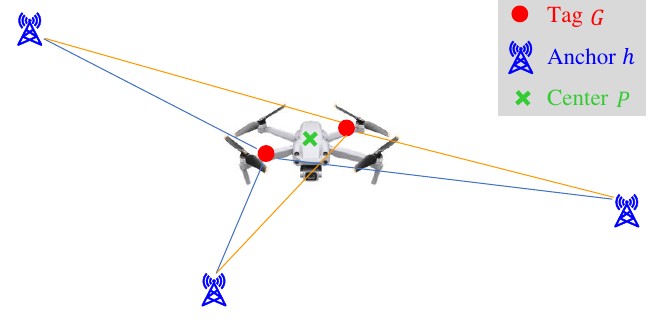}\caption{\label{fig:UWB_Orient} Multiple tag placement for pose (orientation
		+ position) determination.}
\end{figure}

In view of the above discussion, UWB anchors satisfying Assumption
\ref{Assumption:assum_TOA} can position the vehicle. Anchors and
three tags noncollinearly placed on the vehicle are sufficient to
formulate three equations for vehicle coordination analogously to
\eqref{eq:UWB_IMU_UWB} without the need for a magnetometer or an
accelerometer.

\section{Validation \label{sec:UWB_Simulations}}

This section illustrates the robustness of the proposed nonlinear
stochastic filter for inertial navigation based on the fusion of UWB
and IMU sensors. We have used the publicly available real-world drone
flight dataset published by Zhao et al., 2022 \cite{zhao2022uwbData}.
The dataset includes measurements obtained from a drone equipped with
one UWB tag and a 6-axis IMU (see Fig. \ref{fig:UWB_Drone}). 8 fixed
UWB anchors were present during the drone flight (satisfying Assumption
\ref{Assumption:assum_TOA}). Additionally, the dataset supplies ground
truth information (true drone's position and orientation described
with respect to unit-quaternion).
\begin{figure}[h]
	\centering{}\centering\includegraphics[scale=0.22]{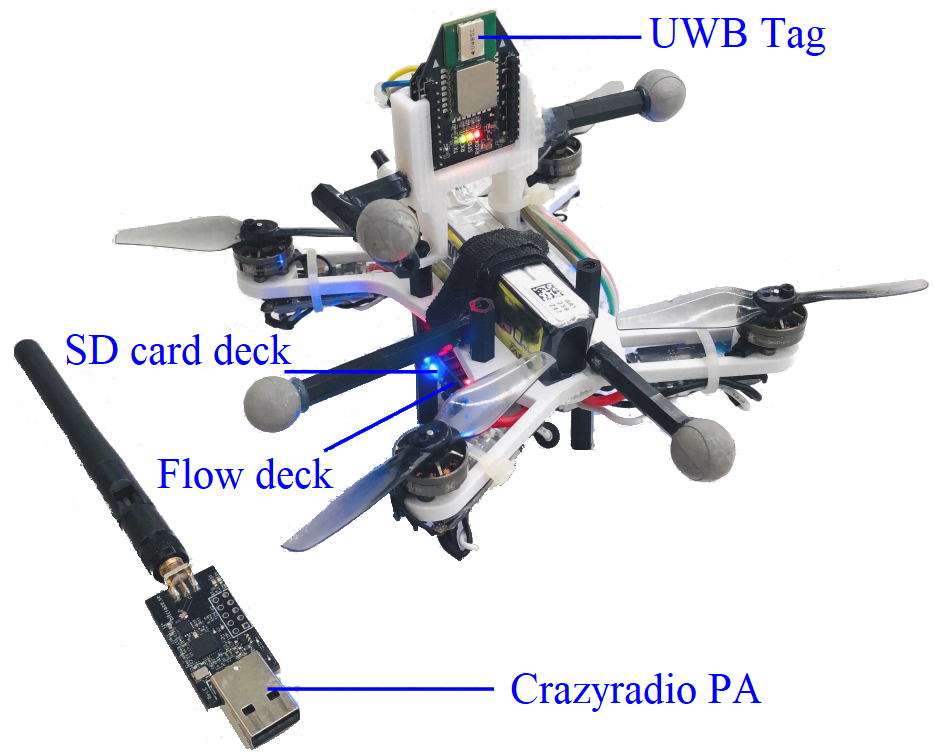}\caption{\label{fig:UWB_Drone} Customized drone with a flight controller \cite{zhao2022uwbData}.}
\end{figure}

The linear velocity is not provided in the dataset. Therefore, a classical
maximum likelihood (ML) method has been utilized to identify the true
linear velocity (for the purpose of comparison) \cite{soderstrom2002discrete,maybeck1982stochastic}.
To test the filter convergence capability against large error initialization,
we initiated the drone flight at its true original position provided
in the dataset $P(0)=[-0.061,1.244,1.506]^{\top}$ and a linear velocity
$V(0)=[-0.4708,0.1308,-0.3363]^{\top}$, while the estimated initial
position and linear velocity were set as $\hat{P}(0)=[-2,-3,0]^{\top}$
and $\hat{V}(0)=[0,0,0]^{\top}$, respectively. The UWB tag is not
positioned at the vehicle's center. As such, in order to adjust the
UWB tag's position to the vehicle's center a translation vector $v_{c}=[-0.012,0.001,0.091]^{\top}m$
\cite{zhao2022uwbData}, the ranging distance is calculated by:
\begin{equation}
	d_{i,j}=||Rv_{c}+P-h_{j}||-||Rv_{c}+P-h_{i}||\label{eq:UWB_dji_simu}
\end{equation}
where $d_{i,j}$ denotes the TDOA range distance, $R$ refers to the
drone's orientation, $P$ stands for the drone's position, and $h_{j}$
denotes the $j$th anchor position (see Section \ref{subsec:TDOA}).
The expression in \eqref{eq:UWB_dji_simu} has been utilized for obtaining
the reconstructed position $P_{y}$ described in \eqref{eq:UWB_Pbar1}.
To include a magnetometer, we set $m_{r}=[-1.3,0,1.5]^{\top}$ and
calculated $m_{m}=R^{\top}m_{r}+n_{m}$ where $n_{m}=\mathcal{N}\left(0,0.2\right)$
refers to a normally distributed random noise vector with a zero mean
and a standard deviation of $0.2$. The design parameters have been
selected as follows: $k_{1}=3$, $k_{v}=3$,\textbf{ }$k_{a}=70$,
$\gamma_{s}=0.1$, $\varepsilon=0.5$, and $k_{\sigma}=0.1$ where
the initial covariance estimate has been set to $\hat{\sigma}(0)=[0,0,0]^{\top}$.
The TDOA of UWB measurements and IMU data were collected at a rate
of 500 Hz. Therefore, to confirm robustness of the proposed filter,
the algorithm was implemented at a lower sampling rate ($100$ Hz)
with $\Delta T=0.01$ sec.

The experimental validation uses ``Const1-Trial'' of the UTIL dataset
\cite{zhao2022uwbData}. Fig. \ref{fig:UWB_Py} illustrates the true
drone's position $P$ marked as a red solid line, the estimated position
$\hat{P}$ plotted as a blue dash line, and the reconstructed position
$P_{y}$ in gray color obtained using $d_{i,j}$ TDOA range distance,
fixed anchors ($h_{i}$ for $i=1,2\ldots,8$), and $v_{c}$ described
in \eqref{eq:UWB_dji_simu}. The reconstructed position $P_{y}$ in
Fig. \ref{fig:UWB_Py} demonstrates a high level of noise attached
to measurements. Despite the high level of measurement uncertainties,
Fig. \ref{fig:UWB_Py} reveals highly accurate estimation of the vehicle's
position in 3D space. Furthermore, Fig. \ref{fig:UWB_Err} illustrates
in blue the successful and rapid error convergence of the proposed
filter in terms of orientation $||\tilde{R}||_{{\rm I}}=\frac{1}{4}{\rm Tr}\{\mathbf{I}_{3}-\hat{R}R^{\top}\}$,
position $||P-\hat{P}||$, and linear velocity $||V-\hat{V}||$ from
large error initialization to the neighborhood of the origin. Additionally,
Fig. \ref{fig:UWB_Err} compares the performance of the proposed stochastic
filter against EKF and UKF. As has been illustrated in Fig. \ref{fig:UWB_Err},  UKF is slower in orientation convergence than the proposed stochastic nonlinear filter. At steady-state, UKF shows more oscillatory behavior when compared to the EKF, and EKF presents more oscillatory behavior than the proposed nonlinear stochastic filter.

\begin{figure}[h]
	\centering{}\centering\includegraphics[scale=0.33]{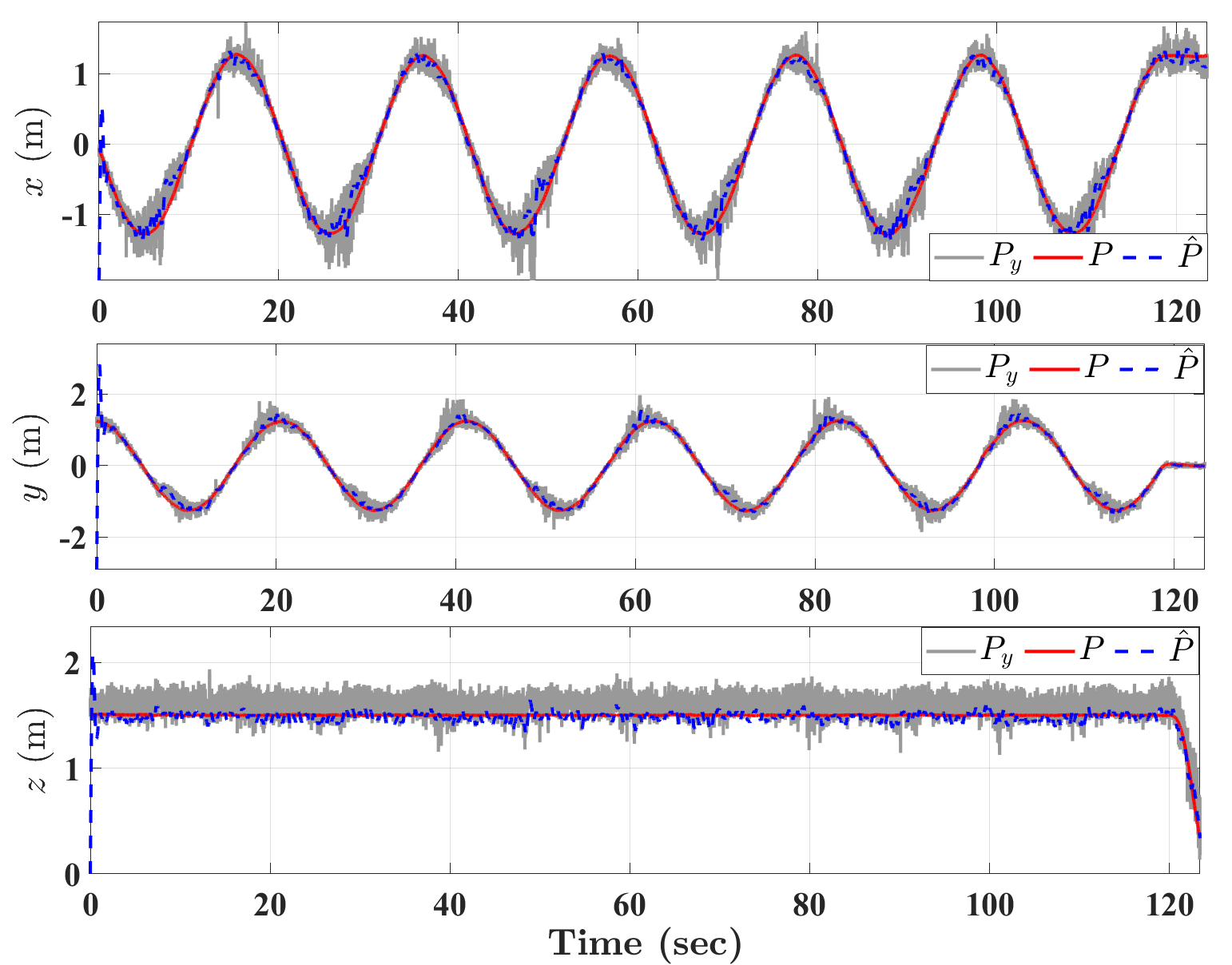}\caption{\label{fig:UWB_Py} Trial Const1 \cite{zhao2022uwbData}: The true
		drone's position $P$ is marked as a red solid line, the estimated
		position $\hat{P}$ (proposed) is plotted as a blue dash line, and
		the reconstructed position $P_{y}$ is presented as a gray solid line.}
\end{figure}

\begin{figure}[h]
	\centering{}\centering\includegraphics[scale=0.33]{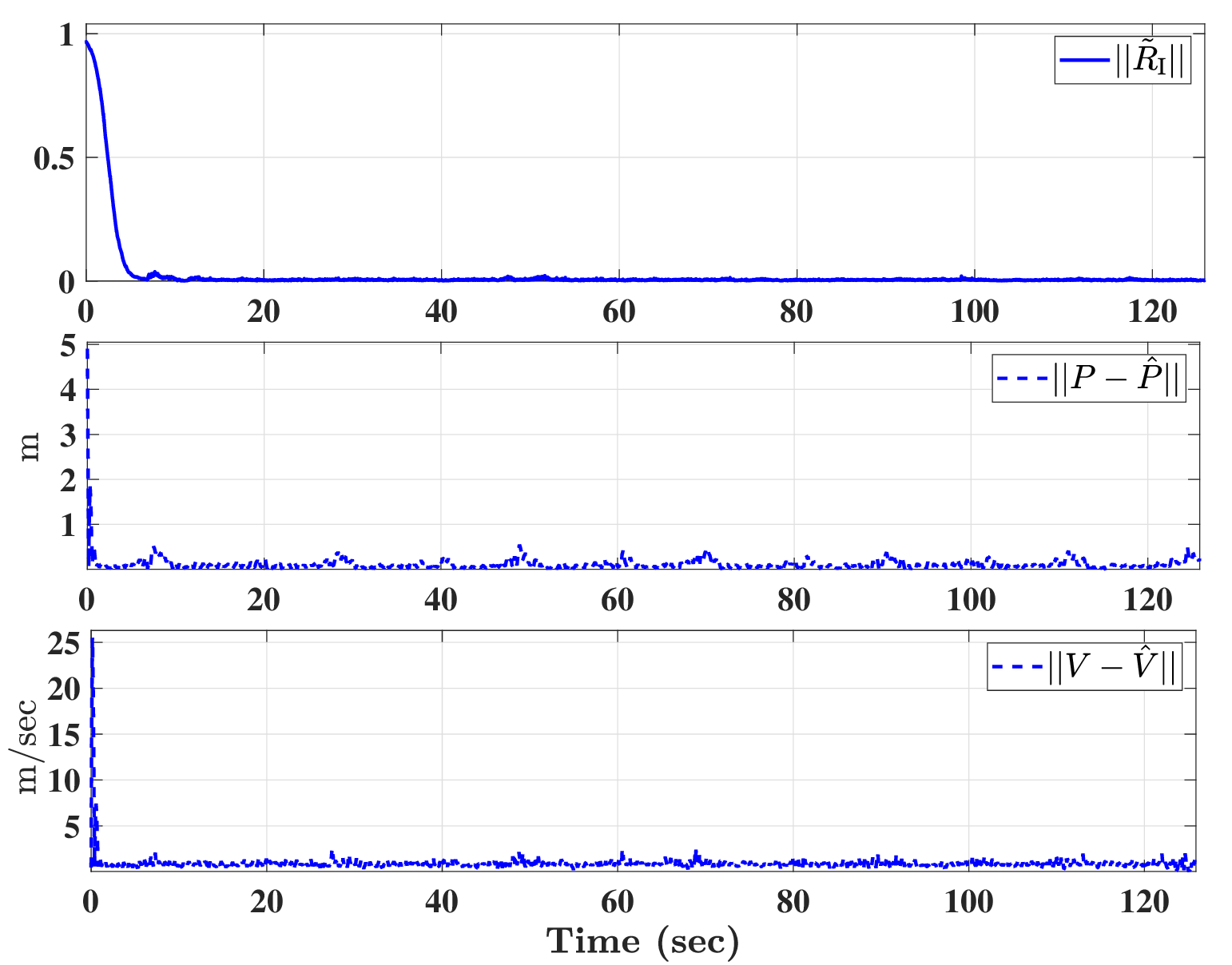}\caption{\label{fig:UWB_Err} Error convergence: Nonlinear stochastic filter
		(proposed) plotted as a blue dash line.}
\end{figure}

\begin{figure}[h]
	\centering{}\textcolor{blue}{\centering\includegraphics[scale=0.3]{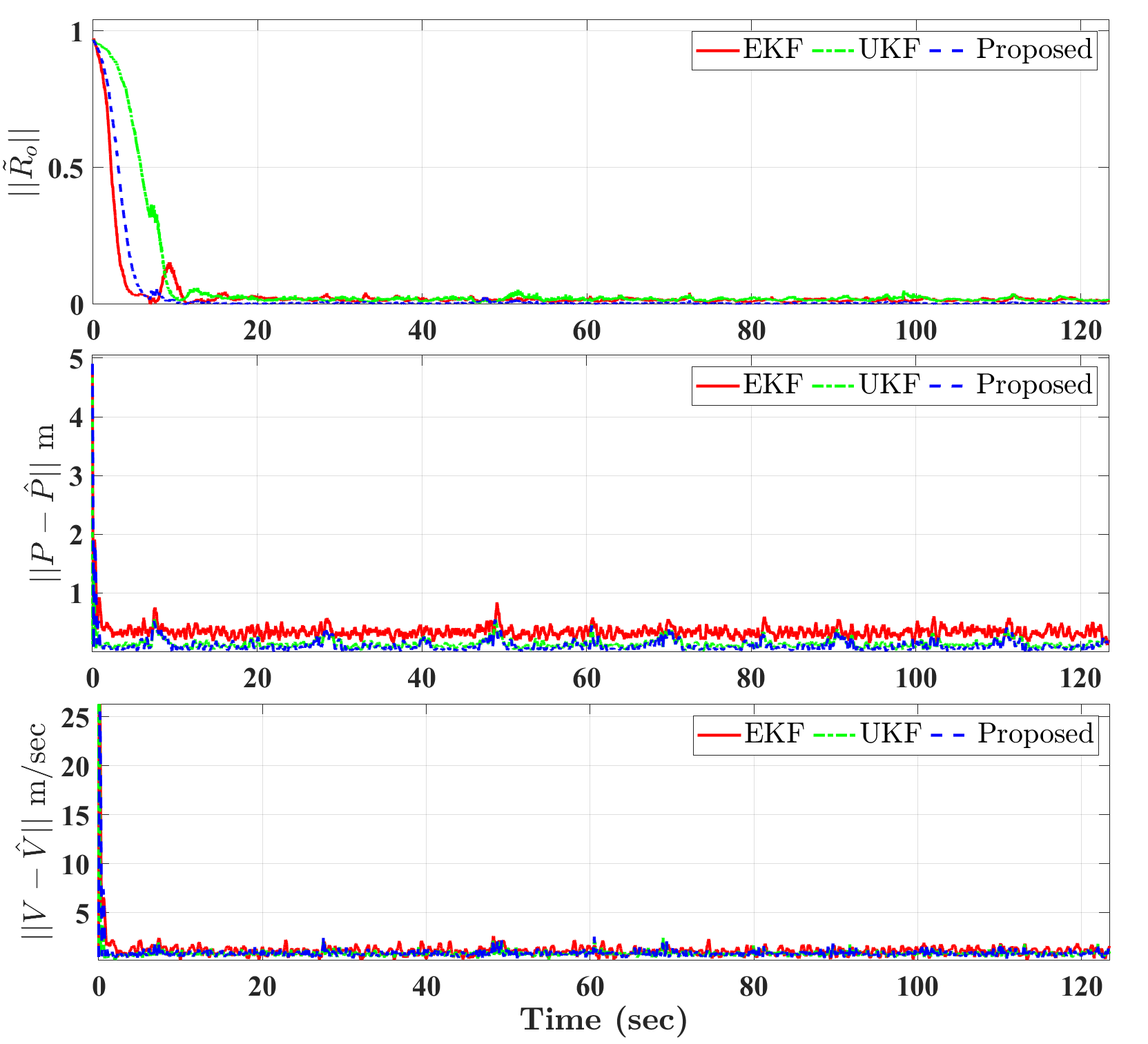}\caption{\label{fig:UWB_Err} Error convergence: Nonlinear stochastic filter
			(proposed) plotted as a blue dash line vs EKF and UKF.}
	}
\end{figure}

\section{Conclusion \label{sec:SE3_Conclusion}}

A novel nonlinear stochastic filter for inertial navigation on the
Lie Group of $\mathbb{SE}_{2}\left(3\right)$ is proposed to estimate
the vehicle's attitude, position, and linear velocity. The proposed
filter utilizes the direct measurements supplied by UWB and IMU achieving
semi-globally uniformly ultimately bounded (SGUUB) stability of the
closed loop error signals. The filter effectively tackles IMU uncertainties
and ensures noise attenuation. The necessary conditions for orientation
and position determination have been outlined considering IMU and
UWB fusion or the exclusive use of UWB technology. The validation
performed using a real-world unmanned aerial vehicle flight dataset
has revealed the capability of the discrete form of the proposed approach
to produce accurate estimates of orientation, position, and linear
velocity.

\section*{Acknowledgment}

The authors would like to thank \textbf{Maria Shaposhnikova} for proofreading
the article.

\section*{Appendix\label{sec:UWB_AppendixA}}
\begin{center}
	\textbf{Quaternion Representation of Navigation Filter}
	\par\end{center}

\noindent Let $Q=[q_{0},q^{\top}]^{\top}\in\mathbb{S}^{3}$ refer
to a unit-quaternion vector with $q_{0}\in\mathbb{R}$ and $q\in\mathbb{R}^{3}$,
and the 3-sphere group $\mathbb{S}^{3}$ be defined by 
\[
\mathbb{S}^{3}=\{\left.Q\in\mathbb{R}^{4}\right|||Q||=\sqrt{q_{0}^{2}+q^{\top}q}=1\}
\]
$\mathbb{SO}\left(3\right)$ can be obtained from unit-quaternion
through the following map $\mathcal{R}_{Q}:\mathbb{S}^{3}\rightarrow\mathbb{SO}\left(3\right)$
\cite{hashim2021_COMP_ENG_PRAC}:
\begin{align}
	\mathcal{R}_{Q} & =(q_{0}^{2}-||q||^{2})\mathbf{I}_{3}+2qq^{\top}+2q_{0}\left[q\right]_{\times}\in\mathbb{SO}\left(3\right)\label{eq:NAV_Append_SO3}
\end{align}
Let $\hat{Q}=[\hat{q}_{0},\hat{q}^{\top}]^{\top}\in\mathbb{S}^{3}$
denote the estimate of $Q=[q_{0},q^{\top}]^{\top}\in\mathbb{S}^{3}$
and the attitude estimate from quaternion estimate be given by
\[
\hat{\mathcal{R}}_{Q}=(\hat{q}_{0}^{2}-||\hat{q}||^{2})\mathbf{I}_{3}+2\hat{q}\hat{q}^{\top}+2\hat{q}_{0}\left[\hat{q}\right]_{\times}\in\mathbb{SO}\left(3\right)
\]
Similar to \eqref{eq:UWB_Py}, let us obtain the vehicle position
as follows:
\begin{equation}
	\begin{cases}
		P_{y}=(A^{\top}A)^{-1}A^{\top}B, & \text{TOA}\\
		\overline{P}_{y}=\left[\begin{array}{c}
			P_{y}\\
			||P-h_{1}||
		\end{array}\right]=(A^{\top}A)^{-1}A^{\top}B, & \text{TDOA}
	\end{cases}\label{eq:UWB_Py-1}
\end{equation}
Consider the covariance adaptation mechanism and the correction factors
using the following set of equations:
\begin{equation}
	\begin{cases}
		E_{r} & =\frac{1}{4}{\rm Tr}\sum_{i=1}^{3}s_{i}(r_{i}r_{i}^{\top}-\hat{\mathcal{R}}_{Q}\hat{v}_{i}v_{i}^{\top}\hat{\mathcal{R}}_{Q}^{\top})\\
		\mathcal{D}_{v} & ={\rm diag}(\sum_{i=1}^{3}s_{i}v_{i}\times\hat{v}_{i})\\
		\dot{\hat{\sigma}} & =\gamma_{\sigma}\frac{E_{r}+2}{8}\exp(E_{r})\mathcal{D}_{{\rm v}}(\sum_{i=1}^{n}s_{i}v_{i}\times\hat{v}_{i})-k_{\sigma}\gamma_{\sigma}\hat{\sigma}\\
		w_{\Omega} & =-\frac{k_{1}}{2}\sum_{i=1}^{n}\hat{\mathcal{R}}_{Q}s_{i}(v_{i}\times\hat{v}_{i})-\frac{1}{8}\frac{E_{r}+2}{E_{r}+1}\hat{\mathcal{R}}_{Q}\mathcal{D}_{{\rm v}}\hat{\sigma}\\
		w_{V} & =-\frac{k_{v}}{\varepsilon}(P_{y}-\hat{P})-[w_{\Omega}]_{\times}\hat{P}\\
		w_{a} & =-k_{a}(P_{y}-\hat{P})-[w_{\Omega}]_{\times}\hat{V}
	\end{cases}\label{eq:UWB_Filter1_Correc-1}
\end{equation}
Consider the navigation stochastic filter kinematics as follows:
\begin{equation}
	\begin{cases}
		\Theta_{m}= & \left[\begin{array}{cc}
			0 & -\Omega_{m}^{\top}\\
			\Omega_{m} & -[\Omega_{m}]_{\times}
		\end{array}\right],\hspace{1em}\Psi=\left[\begin{array}{cc}
			0 & -w_{\Omega}^{\top}\\
			w_{\Omega} & [w_{\Omega}]_{\times}
		\end{array}\right]\\
		\dot{\hat{Q}} & =\frac{1}{2}\Theta_{m}\hat{Q}-\frac{1}{2}\Psi\hat{Q}\\
		\dot{\hat{P}} & =\hat{V}-\left[w_{\Omega}\right]_{\times}\hat{P}-w_{V}\\
		\dot{\hat{V}} & =\hat{\mathcal{R}}a_{m}+\overrightarrow{\mathtt{g}}-\left[w_{\Omega}\right]_{\times}\hat{V}-w_{a}
	\end{cases}\label{eq:UWB_Filter1_Detailed-1}
\end{equation}
\bibliographystyle{IEEEtran}
\bibliography{bib_UWB}
\end{document}